\documentclass[11pt]{article}
\usepackage[a4paper,margin=25mm,headheight=15pt]{geometry}
\usepackage[T1]{fontenc}
\usepackage{lmodern,microtype}
\usepackage{amsmath,amssymb,amsthm,mathtools}
\usepackage{booktabs,array,enumitem,float,caption,tikz,needspace}
\usetikzlibrary{arrows.meta}
\usepackage[dvipsnames]{xcolor}
\usepackage{hyperref,fancyhdr}
\hypersetup{colorlinks=true,linkcolor=MidnightBlue,citecolor=MidnightBlue,urlcolor=MidnightBlue,pdftitle={Euclidean SVP is NP-hard for Cyclic Lattices},pdfauthor={Daqing Wan},pdfsubject={Deterministic exact-SVP hardness for cyclic lattices, with an application to algebraic NTRU-form lattices}}
\setlist{nosep,leftmargin=1.5em}
\makeatletter
\renewcommand*\l@section{\@dottedtocline{1}{0em}{1.8em}}
\makeatother
\allowdisplaybreaks[1]
\newcommand{\stepheading}[1]{\par\medskip\noindent\textbf{#1}\enspace\ignorespaces}
\newtheorem{theorem}{Theorem}[section]
\newtheorem{lemma}[theorem]{Lemma}
\newtheorem{proposition}[theorem]{Proposition}
\theoremstyle{definition}
\newtheorem{definition}[theorem]{Definition}
\newtheorem*{definition*}{Definition}
\theoremstyle{plain}
\newcommand{\Z}{\mathbb Z}\newcommand{\Q}{\mathbb Q}\newcommand{\F}{\mathbb F}\newcommand{\PP}{\mathbb P}\newcommand{\Aff}{\mathbb A}
\newcommand{\supp}{\operatorname{supp}}\newcommand{\OPT}{\operatorname{OPT}}\newcommand{\Sing}{\operatorname{Sing}}
\newcommand{\norm}[1]{\left\lVert#1\right\rVert}\newcommand{\one}{\mathbf 1}

\title{\textbf{Euclidean SVP is NP-hard for Cyclic Lattices}}
\author{Daqing Wan\thanks{This manuscript was developed with ChatGPT assistance. The arithmetic input is due to the author.}}
\date{September 14, 2026}

\begin{document}
\maketitle
\vspace{-6mm}
\begin{abstract}
We prove that exact Euclidean SVP is NP-hard under deterministic polynomial-time many-one reductions for full-rank cyclic integer lattices, equivalently full-rank ideals of $R_N:=\Z[X]/(X^N-1)$ in the coefficient norm. Hardness holds with $N=q-1$ for a varying odd prime $q$. As an application, we prove the same hardness for the algebraic class of NTRU-form lattices
$\{(x,z)\in R_N^2:Hx\equiv z\pmod{QR_N}\}$, where $H,Q$ are unrestricted inputs. The decision problems are NP-complete, and the exact search problems are NP-hard under polynomial-time Turing reductions. No hardness claim is made for cryptographic NTRU parameter subclasses or key-generation distributions.
\end{abstract}
\par\medskip
\section{Introduction}
\label{sec:statement}
The shortest vector problem asks for a shortest nonzero vector in a lattice given by a basis. Its complexity depends not only on the norm and the requested approximation factor, but also on the class of lattices supplied as input. An unrestricted hardness reduction may destroy precisely the algebraic structure that makes a family interesting. Thus NP-hardness for general lattices does not, by itself, establish hardness for lattices invariant under a prescribed symmetry.

We study one of the simplest such families: full-rank integer lattices preserved by cyclic rotation of their coordinates. Equivalently, these are full-rank ideals of the quotient ring $\Z[X]/(X^N-1)$, with the Euclidean norm of coefficient vectors. In his work on generalized compact knapsacks, Micciancio explicitly asked whether the shortest vector problem on cyclic lattices is NP-hard~\cite[Section 1]{Cyclic}. We give a deterministic reduction answering the exact decision version of that question for this coefficient-norm setting. The reduction maintains invariance under a \emph{one-coordinate} cyclic shift, rather than only a block shift or a product of independent shifts.

The main obstacle is that an input assignment refers to named variables, while a cyclic lattice has no preferred coordinate origin. We resolve this by constructing a lattice whose shortest vectors carry their own recognizable origin: one occupied position in one parity class, and many occupied positions in the other. The unique position is a relative anchor. A uniform finite-field counting argument realizes every Boolean assignment relative to an anchor without creating unwanted pair correlations. Those correlations express a quadratic set-cover cost, and multiplication by an explicit ring element converts that cost into ordinary Euclidean length.

\subsection{The problem and the main result}
For a positive integer $N$, let
\[
 R_N=\Z[X]/(X^N-1).
\]
Represent an element by its unique polynomial of degree less than $N$. Its \emph{coefficient Euclidean norm} is
\[
 \norm{f}_2^2=\sum_{a=0}^{N-1}f_a^2.
\]
Let $S$ denote the cyclic shift $(Sy)_a=y_{a-1}$, with indices modulo $N$. Multiplication by $X$ corresponds to $S$.

An additive subgroup $I\subseteq R_N$ is an ideal if and only if it is closed under multiplication by $X$: closure under powers of $X$ and integer linear combinations gives closure under every element of $R_N$. In coefficient coordinates this condition is $SI\subseteq I$. It implies $SI=I$, since $S^N$ is the identity. A basis supplied for a cyclic lattice need not itself be a circulant matrix.

\begin{definition}[Exact cyclic-ideal decision SVP]
The input is a positive integer $N$, a nonsingular integer matrix $B\in\Z^{N\times N}$ whose columns generate a cyclic lattice $I=B\Z^N$, and a nonnegative integer squared threshold $b$. Decide whether
\[
 \lambda_1(I)^2:=\min_{0\ne v\in I}\norm{v}_2^2\le b.
\]
Inputs that do not encode such a triple $(N,B,b)$ are rejected; this includes inconsistent matrix dimensions or a negative threshold. Singular matrices and bases whose column lattices are not cyclic are also rejected. Section~\ref{sec:algorithm} gives polynomial-time validity tests.
\end{definition}

\begin{theorem}[Main result]
\label{thm:main}
Exact cyclic-ideal decision SVP is NP-complete under deterministic polynomial-time many-one reductions. Hardness already holds for dimensions $N=q-1$, with $q$ an odd prime. The output ideal is full rank and is supplied by an explicit additive integer basis. Exact search-SVP on the same class is NP-hard under polynomial-time Turing reductions.
\end{theorem}

The ring, norm, and representation in this statement are important. The quotient $R_N$ is not the ring of integers of a single number field. Over $\Q$ it decomposes as
\[
 R_N\otimes_{\Z}\Q\simeq\prod_{d\mid N}\Q(\zeta_d),
\]
where $\zeta_d$ is a primitive $d$th root of unity. This rational decomposition is not an assertion that $R_N$ is the product of the corresponding integer rings. Our proof works in $R_N$ itself and does not project to one factor. The constructed ideal is not asserted to be principal. Nor is the dimension fixed: the prime $q$, and hence $N$, varies with the input.

The theorem concerns worst-case exact complexity. It gives no average-case hardness theorem, no fixed approximation factor greater than one, and no security proof for a cryptographic scheme. These distinctions separate the present result from several strands of earlier work discussed below.

\paragraph{Application to NTRU form.}
A second result concerns the explicit class
\[
 \Lambda_{H,Q}=\{(x,z)\in R_N^2:Hx\equiv z\pmod{QR_N}\},
\]
with the coefficient Euclidean norm on both blocks, arbitrary $Q\ge2$, and arbitrary $H\in R_N/QR_N$. Theorem~\ref{thm:ntru} proves exact decision NP-completeness for this algebraic class, already with $N=q-1$. The proof adapts our particular hard ideals so that each is the kernel of one modular convolution; a threshold-preserving NTRU embedding then preserves exactly all vectors below the threshold. These are rank-two convolution modules, not asserted to be single-cycle lattices in dimension $2N$. No promise of a cryptographically generated public key is included.

\subsection{Why cyclic lattices are an interesting complexity class}
Cyclic symmetry is algebraically strong but inexpensive to state. A vector brings with it all its cyclic shifts, and multiplication in $R_N$ is circular convolution of coefficient vectors. This makes the class a natural test of whether a large, explicit symmetry can simplify an integer optimization problem. The answer cannot be read off from complex linear algebra: although circulant operators can be diagonalized by a Fourier transform, that transform does not turn the integral feasible set into independent coordinate choices. The divisibility and compatibility conditions defining the lattice remain essential.

The symmetry also relates cyclic lattices to cyclic codes. If a linear code over a prime field is invariant under rotation, its inverse image under reduction modulo that prime is a cyclic integer lattice. Reed--Solomon constructions provide particularly useful moment constraints; Bennett and Peikert showed how their integer lifts yield locally dense lattices for SVP reductions~\cite{BP}. We use a related lift, but add parity constraints so that the shortest vectors themselves, rather than only vectors near a target, have a prescribed encoding shape. Geometric work on cyclic lattices, including their well-rounded subfamilies, gives another perspective on the consequences of the symmetry~\cite{FS}.

Cyclic and ideal lattices are also central to the history of efficient lattice-based cryptography. Micciancio's generalized compact knapsacks connect average-case inversion to worst-case approximation problems on cyclic lattices and reduce the description and evaluation costs of the resulting one-way functions~\cite{Cyclic}. Peikert--Rosen and Lyubashevsky--Micciancio developed collision-resistant constructions based on related cyclic or ideal-lattice assumptions~\cite{PR,LM}. These are worst-case-to-average-case reductions: they transfer an assumed worst-case hardness to cryptographic problems. They do not establish that the underlying structured lattice problem is NP-hard.

The original NTRU construction provides another motivation for studying convolution structure. Its public-lattice representation has two polynomial components, linked by one congruence, and a block-circulant basis~\cite[Sections 1.1--1.2 and 3.4.1]{NTRU}. Our application uses the same NTRU lattice form, but allows arbitrary public polynomials and moduli. It concerns worst-case exact SVP, not recovering a prescribed small secret or distinguishing a key distribution. Pellet-Mary--Stehl\'e's reductions for specified search-NTRU variants~\cite{PMS} address different promises, including an average-case target; those results and the application here do not imply one another merely from the shared lattice shape.

That distinction is particularly relevant here. The assumptions in these constructions involve approximation factors and sometimes restricted dimension families; for example, the Peikert--Rosen cyclic-SVP assumption uses sufficiently large prime dimensions~\cite{PR}. Our hard dimensions are $q-1$, which are even, and our theorem is exact. Neither the dimension restriction nor the approximation requirement of that cryptographic assumption follows from Theorem~\ref{thm:main}. The motivation is to understand the worst-case complexity of an important structured class, not to replace the security analyses of those schemes.

\subsection{General lattices: hardness, approximation, and derandomization}
The early complexity questions for short lattice vectors go back to van Emde Boas~\cite{vEB}. At the algorithmic level, the Lenstra--Lenstra--Lov\'asz algorithm gives a polynomial-time lattice reduction procedure and an exponential-factor approximation to the shortest vector~\cite{LLL}. Such an approximation guarantee is consistent with hardness of finding an exact shortest vector.

Ajtai proved Euclidean SVP NP-hard under randomized reductions in 1998~\cite{Ajtai}. Micciancio subsequently obtained randomized hardness for every fixed approximation factor below $\sqrt2$~\cite{Mic01}. Khot extended randomized approximation hardness to arbitrary fixed constants~\cite{Khot}. Haviv and Regev developed a tensor-based amplification analysis giving stronger dimension-dependent factors under correspondingly stronger complexity assumptions, including a randomized quasipolynomial-time reduction~\cite{HR}. The type and running time of the reduction are part of these statements; the stronger dimension-dependent factors are not being cited as deterministic polynomial-time NP-hardness results.

A major line of work sought to remove randomness from the locally dense lattice method. A locally dense lattice has relatively large nonzero minimum, while a specified small ball centered away from the lattice origin contains many lattice points. To encode assignments, one also needs a projection or prescription guarantee, not just a large cardinality. Micciancio developed this framework and its derandomization goals~\cite{Mic01,Micciancio}. Bennett and Peikert supplied a particularly simple Reed--Solomon construction and identified an explicit route toward deterministic reductions~\cite{BP}.

Independently, Hair and Sahai developed a different deterministic approach based on PCP techniques and a fortified tensor construction~\cite{HairSahai}. For every finite $p$, including the Euclidean case $p=2$, they obtain hardness within $2^{(\log n)^{1-o(1)}}$ under the subexponential-time hypothesis $\mathrm{NP}\not\subseteq\bigcap_{\delta>0}\mathrm{DTIME}(\exp(n^\delta))$. This is an important deterministic Euclidean hardness result, but it is not a polynomial-time NP-hardness theorem.

In earlier work~\cite{Wan}, we proved that exact Euclidean SVP is NP-hard under deterministic polynomial-time reductions, settling the 1981 conjecture of van Emde Boas; the same work gives deterministic NP-hardness for fixed approximation factors below $\sqrt2$. In subsequent work~\cite{WanApprox}, we established that, for every fixed constant $\rho>1$, Euclidean SVP is NP-hard to approximate within factor $\rho$ under deterministic polynomial-time many-one reductions. We cite the September 8 version of~\cite{WanApprox}, whose amplification uses Euclidean length, modular support, and matrix rank. The present proof uses the finite-field counting ideas of~\cite{Wan}, but neither general-lattice NP-hardness theorem nor the tensor amplification of~\cite{WanApprox} is used as an input.

For the present purpose, the missing ingredient in a general-lattice reduction is not hardness of the source problem. It is preservation of cyclic structure. Appending an arbitrary target coordinate, an arbitrary constraint matrix, or independent tensor coordinates generally breaks the required one-step rotation invariance. The construction below implements the constraints using scalar cyclic correlations instead.

\subsection{Structured lattices and the place of this result}
There are several different meanings of algebraically structured lattice. Cyclic coefficient lattices correspond to ideals of $\Z[X]/(X^N-1)$. Ideals in a number field are usually studied through the canonical embedding, which uses all archimedean embeddings of that field. Module lattices allow several algebraic coordinates over a base ring. A bound or reduction for one of these classes does not automatically apply to another.

Ring-LWE connects efficient ring-based cryptography with worst-case problems on number-field ideal lattices~\cite{LPR}. Module-SIS and Module-LWE extend this perspective to module lattices, interpolating between ideal structure and less constrained lattice problems~\cite{LS}. Again, these security reductions must not be identified with an NP-hardness classification of exact SVP. Likewise, two module generators over a degree-$d$ number field describe an ordinary lattice of rank $2d$, not a two-dimensional lattice.

Liu--Feng--Pan prove exact decision-SVP NP-complete under deterministic reductions for full-rank free rank-two modules over $\Z[\zeta_q]$, with $q\equiv3\pmod4$ and the canonical Euclidean norm~\cite{LFP}. Their proof controls the arbitrary ring multiplier introduced by the second module coordinate.

Martin~\cite[Theorem 1.1]{Martin} gives a dimension-preserving deterministic polynomial-time reduction from general Euclidean SVP to ideal-SVP in the canonical embedding of ideals coprime to the conductor, hence invertible, in monogenic orders of totally real number fields. The field and order depend on the input, and the order need not be the full ring of integers. This differs from the prescribed cyclic quotient-ring family and coefficient norm studied here.

The table distinguishes the lattice classes, Euclidean norms, and reduction guarantees.
\begin{center}
\small
\renewcommand{\arraystretch}{1.22}
\begin{tabular}{@{}>{\raggedright\arraybackslash}p{0.30\linewidth}>{\raggedright\arraybackslash}p{0.24\linewidth}>{\raggedright\arraybackslash}p{0.40\linewidth}@{}}
\toprule
Lattice class & Problem and norm & Hardness statement\\
\midrule
Free rank-two prime-cyclotomic modules & Exact SVP; canonical Euclidean & Deterministic decision NP-completeness~\cite{LFP}.\\
Invertible ideals in varying totally real monogenic orders & Exact SVP; canonical Euclidean & Deterministic dimension-preserving reduction from general SVP~\cite{Martin}.\\
Full-rank ideals of $\Z[X]/(X^N-1)$, $N=q-1$ & Exact SVP; coefficient Euclidean & Deterministic decision NP-completeness, Theorem~\ref{thm:main}.\\
Algebraic NTRU-form lattices in $R_N^2$, arbitrary $H,Q$ & Exact SVP; product coefficient Euclidean & Deterministic decision NP-completeness, Theorem~\ref{thm:ntru}; no key-generation promise.\\
\bottomrule
\end{tabular}
\end{center}

Our result addresses the prescribed cyclic coefficient class directly, rather than specializing either of these reductions. The parity conditions exploit the two evaluations $y(1)$ and $y(-1)$ available in the reducible ring. Projecting to a single cyclotomic field can erase those conditions. Hardness over full rings of integers in prescribed field families, and principality of the cyclic output, are not established here.

\subsection{Technical contribution and proof organization}
The reduction starts from Gap Exact Set Cover. We associate with an integer assignment a repair cost that charges both for the selected sets and for errors in covering the universe. A simple repair argument shows that this cost is at least the size of a minimum ordinary set cover. Thus a promised YES instance has a Boolean exact-cover witness of small repair cost, whereas on a promised NO instance every integer assignment has large repair cost. Section~\ref{sec:source} gives the precise definition and its quadratic expansion.

The first structural ingredient is a homogeneous cyclic ideal whose shortest vectors have a recognizable form. Two parity sums and a collection of finite-field moment congruences force every shortest vector, up to sign and cyclic rotation, to be binary with a unique occupied coordinate in one parity class and a prescribed number of occupied coordinates in the other. The unique coordinate serves as a relative anchor. Because the anchor belongs to the vector rather than to an absolute coordinate position, this encoding is compatible with cyclic symmetry.

The second ingredient is a uniform completion theorem. After prescribing the coordinates that encode a Boolean assignment relative to the anchor, one must complete the remaining support while satisfying all moment equations and without introducing unwanted pairs at the displacements later inspected by the checker. The required uniformity comes from a finite-field point-counting argument: the Deligne--Hooley--Katz estimate controls the error exponent~\cite{Deligne,Hooley}, while the Wan--Zhang Betti-number bounds for affine complete intersections give a uniform bound on the total cohomological contribution~\cite{WZ}. These estimates yield one- and two-label inclusion bounds strong enough for a union bound excluding all unwanted checked pairs. The resulting clean completion exists for every Boolean assignment, although the reduction never needs to compute it.

The third ingredient is a shift-invariant quadratic checker. Cyclic correlations at carefully chosen displacements recover the encoded linear variables and dominate the quadratic monomials occurring in the repair cost. This gives a checker that is sound for every shortest vector and agrees exactly with the repair cost on a clean YES completion. The checker is represented by a symmetric circulant quadratic form. We then use this checker as a perturbation of a large scalar map: on the minimum shell the scalar part contributes the same baseline to every vector, while the cross term records the checker value; choosing the scalar sufficiently large controls the remaining quadratic error and keeps longer vectors away from the threshold.

The base ideal also admits a useful algebraic description as the kernel of one cyclic convolution modulo one integer. This description gives a polynomial-time basis construction by standard integer normal-form algorithms~\cite{KB} and is reused in the NTRU application. There, the final length conversion is chosen to preserve a similar modular-kernel description, which can then be embedded into an algebraic NTRU-form lattice without changing the set of vectors below the relevant SVP threshold.

Section~\ref{sec:source} gives a formula-level outline of the reduction. Sections~\ref{sec:parameters} and~\ref{sec:anchor} construct the separated displacements and the parity-anchor ideal. Section~\ref{sec:count} proves the uniform finite-field counting theorem, and Section~\ref{sec:clean} establishes clean completions. Section~\ref{sec:checker} derives the cyclic checker, Section~\ref{sec:conversion} converts the checker gap into Euclidean length, and Section~\ref{sec:algorithm} constructs the output basis and proves polynomial-time complexity. Section~\ref{sec:ntru} applies the construction to the unrestricted algebraic class of NTRU-form lattices.

\section{Outline of the main proof}
\label{sec:source}
This section presents the reduction from its source promise to the final integer basis. The constructions are explicit; the proofs of the geometric and counting claims follow in Sections~\ref{sec:parameters}--\ref{sec:algorithm}.

\subsection{Gap Exact Set Cover and its repair cost}
Let $U=\{e_1,\ldots,e_m\}$ be a universe and $S_1,\ldots,S_r\subseteq U$ the available subsets, with
\[
 \bigcup_{i=1}^r S_i=U.
\]
We may assume this union condition: otherwise replace the input by the fixed coverable NO instance consisting of a two-element universe, its two singleton subsets, and threshold $1$. Its ordinary-cover optimum is $2>3/2$.
Let $A\in\{0,1\}^{m\times r}$ be the incidence matrix, so $A_{ei}=1$ exactly when $e\in S_i$. Assume $m,r\ge1$ and let $1\le\tau\le r$ be an integer.

An ordinary cover has a binary indicator $x$ satisfying $Ax\ge\one_m$ coordinatewise. An exact cover satisfies $Ax=\one_m$: every universe element is covered exactly once. Write
\[
 \OPT(A)=\min\{\norm{x}_1:x\in\{0,1\}^r,\ Ax\ge\one_m\}
\]
for the \emph{ordinary set-cover optimum}: the minimum number of available subsets whose union is $U$. The abbreviation $\OPT$ stands for \emph{optimum}.

We use Gap Exact Set Cover with the fixed factor $3/2$:
\[
\begin{array}{ll}
\text{YES:}&\exists\xi\in\{0,1\}^r:\ A\xi=\one_m,\quad\norm{\xi}_1\le\tau;\\[1mm]
\text{NO:}&\OPT(A)>\tfrac32\tau.
\end{array}
\]
This promise problem is NP-hard by Micciancio~\cite[Definition 2.5 and Theorem 2.6]{Micciancio}. The reduction below uses only this source hardness theorem. A NO instance may have an exact cover of large size, so the checker must charge for both the chosen subsets and the coverage errors.

\begin{lemma}[Unit-penalty repair cost]
\label{lem:repair}
For every $z\in\Z^r$,
\[
 \Psi_A(z):=\norm{z}_1+\norm{Az-\one_m}_2^2\ \ge\ \OPT(A).
\]
\end{lemma}
\begin{proof}
Choose the subsets indexed by $\supp(z)$, and let
\[
 U_0=U\setminus\bigcup_{i\in\supp(z)}S_i
\]
be the elements they leave uncovered. Adding one available subset for each element of $U_0$ produces a cover of size at most $|\supp(z)|+|U_0|$.

Every nonzero integer coefficient has absolute value at least one, so $|\supp(z)|\le\norm{z}_1$. For $e\in U_0$, all columns in the support have zero in row $e$, while the other coefficients of $z$ vanish. Thus $(Az)_e=0$, and that row contributes one to $\norm{Az-\one_m}_2^2$. Therefore
\[
 \OPT(A)\le|\supp(z)|+|U_0|
 \le\norm{z}_1+\norm{Az-\one_m}_2^2.
\]
The reasoning uses the support of $z$, not the signs of its entries.
\end{proof}

The source promise therefore gives the following one-way implications:
\[
\begin{array}{ll}
\text{source YES}&\Longrightarrow
\text{some binary exact-cover witness }\xi\text{ has }\Psi_A(\xi)\le\tau,\\[1mm]
\text{source NO}&\Longrightarrow
\Psi_A(z)>\tfrac32\tau\quad\text{for every }z\in\Z^r.
\end{array}
\]
For the first implication the residual is zero; the second uses Lemma~\ref{lem:repair}.
Thus, if
\[
 \mu_\Psi(A):=\min_{z\in\Z^r}\Psi_A(z),
\]
then the identity map from the source instance gives the NP-hard promise problem
\[
 \boxed{\mu_\Psi(A)\le\tau\quad\text{versus}\quad
        \mu_\Psi(A)>\tfrac32\tau.}
\]
We do not claim the converse implication from a small repair cost to a small exact cover on arbitrary inputs outside the source promise. Only the displayed YES/NO implications are used by the reduction. In the lattice construction, completeness still uses the Boolean exact-cover witness of the original source instance. We do not claim a reduction from every YES instance of the unrestricted repair-cost promise to the checker promise.

For later use, expand the repair cost in coordinates. Since $A$ is the incidence matrix,
\[
 (A^{\mathsf T}A)_{ii}=|S_i|,\qquad
 (A^{\mathsf T}A)_{ij}=|S_i\cap S_j|\ (i\ne j),\qquad
 (\one_m^{\mathsf T}A)_i=|S_i|.
\]
Therefore, for every $z\in\Z^r$,
\[
\begin{aligned}
 \Psi_A(z)
 ={}&m+\sum_{i=1}^r |z_i|-2\sum_{i=1}^r|S_i|z_i
      +\sum_{i=1}^r|S_i|z_i^2\\
 &\quad+2\sum_{1\le i<j\le r}|S_i\cap S_j|z_i z_j.
\end{aligned}
\]
The vectors decoded from shortest lattice vectors below will be nonnegative. On $z\in\Z_{\ge0}^r$, the absolute values disappear, and the repair cost is the explicit quadratic polynomial
\[
\boxed{\begin{aligned}
 \Psi_A(z)
 ={}&m+\sum_i(1-2|S_i|)z_i+\sum_i|S_i|z_i^2\\
 &\quad+2\sum_{i<j}|S_i\cap S_j|z_i z_j.
\end{aligned}}
\]
This formula tells us exactly what the lattice construction must reproduce: the linear quantities $z_i$, the squares $z_i^2$, and the cross-products $z_i z_j$.

\paragraph{Remark (an alternative source for exact hardness).}
Exact Cover by 3-Sets also suffices for the exact theorem; its NP-hardness follows, for example, from three-dimensional matching~\cite{Karp}. On a coverable instance with $|U|=3\tau$ and all $|S_i|=3$, any ordinary cover using at most $\tau$ subsets must use exactly $\tau$ pairwise disjoint triples, and hence is exact. Thus a NO instance has $\OPT(A)\ge\tau+1$, while a YES exact cover has repair cost $\tau$. Lemma~\ref{lem:repair} then supplies the unit additive gap used in Sections~\ref{sec:conversion} and~\ref{sec:ntru}. We retain Gap Exact Set Cover, as used in~\cite{WanApprox}, throughout the proof: its stronger multiplicative source gap may be useful for future approximation results. No constant-factor SVP hardness is inferred here from that source gap.

\subsection{The construction in explicit form}
\label{sec:roadmap}
We now describe the reduction in the order in which each new object is forced by the preceding one. The purpose of the later sections is to prove the claims previewed here. Our final output will be a basis $\mathbf B_{\rm out}$ of a cyclic ideal and an integer squared threshold $B_*$ such that, on promised source inputs,
\[
 \boxed{\text{source YES}\iff
 \exists\,0\ne v\in\mathbf B_{\rm out}\Z^{q-1}:\ \norm{v}_2^2\le B_*.}
\]

\Needspace{12\baselineskip}
\stepheading{1. Build a cyclic ideal with a rigid shortest-vector shape.}
Choose the parameters of Section~\ref{sec:parameters},
\[
 k=6(r+3),\qquad h=\frac{3k}{2},\qquad
 (100k)^{12}<q<2(100k)^{12},
\]
with $q$ prime and $\alpha\in\F_q^\times$ primitive. All positions are residues modulo $q-1$, and the ring is
\[
 R_{q-1}=\Z[X]/(X^{q-1}-1).
\]
For $y\in\Z^{q-1}$ define the parity sums
\[
 u(y)=\sum_{a\text{ even}}y_a,\qquad
 v(y)=\sum_{a\text{ odd}}y_a,
\]
and define
\[
\boxed{
 \mathcal I=\left\{y\in\Z^{q-1}:
 \begin{array}{l}
 hu(y)-v(y)\equiv0\pmod{h^2-1},\\[1mm]
 \displaystyle\sum_{a=0}^{q-2}y_a\alpha^{aj}=0\text{ in }\F_q
 \quad(1\le j<k)
 \end{array}\right\}.}
\]
A cyclic rotation preserves these conditions, so $\mathcal I$ is a cyclic lattice; Section~\ref{sec:anchor} proves that it is a full-rank ideal. It also proves
\[
 \norm{y}_2^2\ge h+1\qquad(0\ne y\in\mathcal I),
\]
and that equality forces, after an overall sign change and a rotation, a binary vector of the schematic form
\begin{center}
\renewcommand{\arraystretch}{1.2}
\begin{tabular}{@{}lcccccccc@{}}
\toprule
Position & $0$ & $1$ & $2$ & $3$ & $4$ & $5$ & $\cdots$ & $q-2$\\
Entry & $\boxed{1}$ & $*$ & $0$ & $*$ & $0$ & $*$ & $\cdots$ & $*$\\
\bottomrule
\end{tabular}
\end{center}
Exactly $h$ odd positions are occupied, and the boxed entry is the unique occupied even position. We call it the \emph{anchor}. The lattice itself has no distinguished origin; we rotate a shortest vector so that its anchor is at $0$ only for analysis. The field equations still constrain the $h$ occupied odd positions.

\Needspace{13\baselineskip}
\stepheading{2. Let the shortest vector decode $z$, and let $\Psi_A$ dictate the required displacements.}
Suppose $y$ is a shortest vector normalized as above. Choose odd residues $d_1,\ldots,d_r$ such that the positions $\pm d_i$ are distinct. For the moment, no further property of the $d_i$ is needed.

Define the cyclic correlation
\[
 C_y(d)=\sum_{a\bmod(q-1)}y_a y_{a+d}.
\]
Because $d_i$ is odd, every pair counted by $C_y(d_i)$ has one even endpoint. The anchor at $0$ is the only occupied even position, so the only possible nonzero terms are the pairs $(0,d_i)$ and $(-d_i,0)$. Hence
\[
 C_y(d_i)=y_{d_i}+y_{-d_i}.
\]
This identity tells us how a shortest lattice vector should encode the $i$-th source coordinate: define
\[
 \boxed{z_i:=y_{d_i}+y_{-d_i}=C_y(d_i)\in\{0,1,2\}.}
\]
Thus $z$ is not itself a vector of the cyclic ideal; it is the $r$-dimensional nonnegative integer vector decoded from a shortest vector $y\in\mathcal I$. The decoded vector need not be Boolean. Later, Proposition~\ref{prop:clean} proves that every Boolean vector $\xi\in\{0,1\}^r$ has a clean minimum realization with $y_{d_i}=\xi_i$ and $y_{-d_i}=0$, and hence decodes exactly to $z=\xi$.

Now ask what is required to evaluate the quadratic polynomial $\Psi_A(z)$. Since the two entries $y_{d_i},y_{-d_i}$ are binary,
\[
 z_i^2=z_i+2y_{d_i}y_{-d_i}.
\]
The extra product uses the pair $(-d_i,d_i)$, whose displacement is $2d_i$, so
\[
 y_{d_i}y_{-d_i}\le C_y(2d_i).
\]
For $i<j$,
\[
\begin{aligned}
 z_i z_j
 ={}&y_{d_i}y_{d_j}+y_{-d_i}y_{-d_j}
     +y_{-d_i}y_{d_j}+y_{d_i}y_{-d_j}.
\end{aligned}
\]
The first two products occur among the terms of $C_y(d_j-d_i)$, while the last two occur among the terms of $C_y(d_i+d_j)$. Therefore
\[
 z_i z_j\le C_y(d_j-d_i)+C_y(d_i+d_j).
\]
Consequently the quadratic polynomial itself tells us exactly which even pair displacements we must control:
\[
 \boxed{\mathcal D
 =\{2d_i\}_i\cup\{d_j-d_i:i<j\}\cup\{d_i+d_j:i<j\}.}
\]
This is the reason for introducing $\mathcal D$. The displacements $2d_i$ represent the correction in $z_i^2$, while $d_j-d_i$ and $d_i+d_j$ represent the four products in $z_i z_j$.

We choose the $d_i$ so that all members of $\mathcal D$ are distinct and do not wrap around modulo $q-1$. Otherwise one correlation could be forced to serve two unrelated quadratic monomials. Section~\ref{sec:parameters} gives the explicit choice
\[
 \eta_i=i+(2r+1)i^2,\qquad B_0=10\eta_r+10,\qquad
 d_i=2B_0+2\eta_i+1,
\]
and proves the required separation.

\Needspace{14\baselineskip}
\stepheading{3. Replace the repair polynomial by a cyclic checker, and isolate the checker promise.}
Substituting $z_i^2=z_i+2y_{d_i}y_{-d_i}$ into the nonnegative quadratic expansion above gives, for a normalized shortest vector,
\[
\boxed{\begin{aligned}
 \Psi_A(z)={}&m+\sum_i(1-|S_i|)z_i
       +2\sum_i|S_i|y_{d_i}y_{-d_i}\\
 &\quad+2\sum_{i<j}|S_i\cap S_j|z_i z_j.
\end{aligned}}
\]
Now replace each source-side quantity by the cyclic correlation that reads it or bounds it:
\[
\begin{array}{c|c}
\text{quantity in }\Psi_A(z)&\text{cyclic correlation}\\
\hline
z_i&C_y(d_i)\quad\text{exactly},\\
y_{d_i}y_{-d_i}&C_y(2d_i)\quad\text{an upper bound},\\
z_i z_j&C_y(d_j-d_i)+C_y(d_i+d_j)\quad\text{an upper bound}.
\end{array}
\]
This forces the definition
\[
\boxed{\begin{aligned}
 E_A(y)={}&m+\sum_i(1-|S_i|)C_y(d_i)
       +2\sum_i|S_i|C_y(2d_i)\\
 &\quad+2\sum_{i<j}|S_i\cap S_j|
          \bigl(C_y(d_j-d_i)+C_y(d_i+d_j)\bigr).
\end{aligned}}
\]
The possibly negative coefficient $1-|S_i|$ multiplies an exact identity; every replacement used only as an upper bound has a nonnegative coefficient. Hence every shortest vector, after normalization, decodes a nonnegative $z$ with
\[
 \boxed{E_A(y)\ge\Psi_A(z).}
\]
Because correlations are invariant under cyclic rotation and overall sign, the inequality holds for every shortest vector without choosing an absolute anchor in the definition of $E_A$.

This gives the soundness half of the intermediate checker problem. If the source instance is NO, Lemma~\ref{lem:repair} gives
\[
 E_A(y)\ge\Psi_A(z)\ge\OPT(A)>\tfrac32\tau
\]
for every shortest $y\in\mathcal I$.

For completeness, let $\xi$ be a Boolean exact-cover witness. We want a shortest $y$ for which the above correlation bounds are equalities and whose decoder equals $\xi$. Prescribe
\[
 y_{d_i}=\xi_i,\qquad y_{-d_i}=0,
\]
and require
\[
 C_y(2d_i)=0,\qquad
 C_y(d_j-d_i)=\xi_i\xi_j,\qquad
 C_y(d_i+d_j)=0.
\]
Such a vector is a \emph{clean completion}. Sections~\ref{sec:count}--\ref{sec:clean} prove that every Boolean $\xi$ has a clean completion in $\mathcal I$ with squared norm $h+1$. For this vector, $z=\xi$ and all correlation replacements are equalities, so
\[
 E_A(y)=\Psi_A(\xi)=\norm{\xi}_1\le\tau.
\]
Thus the source problem has been reduced to the following checker gap on the minimum shell:
\[
\boxed{\begin{array}{ll}
\text{YES:}&\text{some shortest }y\in\mathcal I\text{ has }E_A(y)\le\tau,\\[1mm]
\text{NO:}&\text{every shortest }y\in\mathcal I\text{ has }E_A(y)>\tfrac32\tau.
\end{array}}
\]
Section~\ref{sec:checker} proves this formally. Since $E_A(y)$ is integer-valued, the NO conclusion also implies $E_A(y)\ge\tau+1$.

\Needspace{13\baselineskip}
\stepheading{4. Represent the checker itself by a symmetric circulant operator.}
The checker gap is still only a scalar statistic attached to shortest vectors. To make it affect Euclidean length, first write the checker as an actual quadratic form.

Collect its correlation coefficients as
\[
 E_A(y)=m+\sum_{d\in\Z/(q-1)\Z}\beta_d C_y(d),
\]
where $\beta_d=0$ at every unlisted displacement. Let $S$ be the cyclic shift matrix. Since
\[
 C_y(d)=\frac12 y^{\mathsf T}(S^d+S^{-d})y,
\]
the nonconstant part of $E_A$ is already a symmetric circulant quadratic form. On the minimum shell, $\norm{y}_2^2=h+1$, so the constant can also be homogenized:
\[
 m=\frac{m}{h+1}\norm{y}_2^2.
\]
Multiplying by $2(h+1)$ clears the denominator and leads naturally to
\[
 \boxed{K:=2mI_{q-1}+(h+1)\sum_d\beta_d(S^d+S^{-d}).}
\]
Thus $K$ is not a separate gadget: it is simply the integer symmetric circulant operator representing the checker on the minimum shell. By construction,
\[
 \boxed{y^{\mathsf T}Ky=2(h+1)E_A(y)
 \qquad\text{whenever }\norm{y}_2^2=h+1.}
\]
Equivalently, because $S$ is multiplication by $X$, $K$ is multiplication by the ring element
\[
 k(X)=2m+(h+1)\sum_d\beta_d(X^d+X^{-d})\in R_{q-1}.
\]

\Needspace{15\baselineskip}
\stepheading{5. Turn the checker quadratic form into ordinary Euclidean length by multiplying the ideal.}
If we were allowed to change the norm on $\mathcal I$, the identity above suggests using the symmetric circulant matrix representing the checker as a perturbation of a large scalar map. Standard SVP, however, uses the ordinary Euclidean norm. We therefore change the lattice rather than the norm. On the minimum shell the scalar part contributes the same baseline to every vector, the cross term records the checker value, and a sufficiently large scalar controls the remaining quadratic error while keeping longer vectors away from the threshold.

Choose a large positive integer $M$ and define the ring multiplier
\[
 \boxed{g(X):=M+k(X).}
\]
The final ideal is defined directly by
\[
 \boxed{I':=g(X)\mathcal I.}
\]
In coefficient coordinates, multiplication by $g(X)$ is the $(q-1)\times(q-1)$ matrix
\[
 T=MI_{q-1}+K,
\]
so $I'=T\mathcal I$. This is the meaning of the image construction: the perturbed length $y\mapsto\norm{Ty}_2$ on the old ideal becomes the ordinary Euclidean length of the image vector $v=Ty$ in the new ideal.

For $v=Ty$,
\[
 \norm{v}_2^2
 =M^2\norm{y}_2^2+2M y^{\mathsf T}Ky+\norm{Ky}_2^2.
\]
Hence on the minimum shell,
\[
 \boxed{\norm{v}_2^2
 =(h+1)\bigl(M^2+4M E_A(y)\bigr)+\norm{Ky}_2^2.}
\]
The checker value now appears linearly in an ordinary squared Euclidean length. The term $M^2\norm{y}_2^2$ keeps vectors outside the minimum shell separated, while $\norm{Ky}_2^2$ is a controlled error.

For an explicit bound, set
\[
 L=1+\max_i\sum_j|K_{ij}|,
\qquad
 M=1+\max\{L^2,\ 2L(h+2)+(4\tau+2)(h+1)\}.
\]
Then $\norm K_{\rm op}<L$, and Theorem~\ref{thm:conversion} uses
\[
 \boxed{B_*=(h+1)\bigl(M^2+(4\tau+2)M\bigr).}
\]
The three required bounds are
\[
\begin{array}{ll}
\text{chosen YES shell vector:}&\norm{Ty}_2^2<B_*,\\
\text{every NO shell vector:}&\norm{Ty}_2^2>B_*,\\
\text{every }y\text{ with }\norm{y}_2^2\ge h+2:&\norm{Ty}_2^2>B_*.
\end{array}
\]
The last line uses $\norm{Ty}_2\ge(M-L)\norm{y}_2$ and does not require any sign information about $E_A$ away from the minimum shell.

\Needspace{10\baselineskip}
\stepheading{6. Output an explicit basis of the final ideal.}
Lemma~\ref{lem:baseann} describes $\mathcal I$ as the kernel of a single cyclic convolution modulo $q(h^2-1)$. Section~\ref{sec:algorithm} computes an integer basis $\mathbf B_{\mathcal I}$ of that kernel by Smith normal form. Since multiplication by $g(X)$ has coefficient matrix $T$, the final basis is
\[
 \boxed{\mathbf B_{\rm out}=T\mathbf B_{\mathcal I}.}
\]
The reduction outputs
\[
 \boxed{(\mathbf B_{\rm out},B_*).}
\]
The assignment $\xi$, its clean completion $y$, and the decoded vector $z$ occur only in the proof of correctness; the reduction never computes any of them. The deterministic computation is
\[
 (A,\tau)\longmapsto
 (\mathcal I,E_A,K,g)\longmapsto
 (T\mathbf B_{\mathcal I},B_*).
\]
Sections~\ref{sec:parameters}--\ref{sec:conversion} prove the mathematical implications in this chain, and Section~\ref{sec:algorithm} proves polynomial-time construction, polynomial encoding length, and membership of the target problem in NP.

\section{Parameters and separated offsets}
\label{sec:parameters}
Set
\begin{equation}
\label{eq:parameters}
 k=6(r+3),\qquad h=\frac{3k}{2},\qquad
 (100k)^{12}<q<2(100k)^{12},
\end{equation}
where $q$ is prime, and choose a primitive element $\alpha\in\F_q^\times$.
For the rest of the construction the ring is
\[
 R_{q-1}=\Z[X]/(X^{q-1}-1).
\]
All coefficient indices and cyclic displacements are taken modulo $q-1$. In particular, the dimension is even, $h+1<2k$, and $q>h^2$.

These choices are deterministic polynomial time in the source input size. Bertrand's theorem provides the prime; a search using trial division suffices because $q$ is a fixed polynomial in $r$. Testing candidates by enumerating their powers finds $\alpha$ in polynomial time as well. No algorithm for factoring arbitrary binary integers is assumed.

For $1\le i\le r$, define
\begin{equation}
\label{eq:offsets}
 \eta_i=i+(2r+1)i^2,\qquad B_0=10\eta_r+10,\qquad
 d_i=2B_0+2\eta_i+1.
\end{equation}
These $d_i$ are odd. As explained in Section~\ref{sec:roadmap}, the quadratic monomials in $\Psi_A$ require exactly the even pair displacements
\begin{equation}
\label{eq:D}
 \mathcal D=\{2d_i\}_i\ \cup\ \{d_j-d_i:i<j\}\ \cup\ \{d_i+d_j:i<j\}.
\end{equation}
The purpose of the present section is only to choose the offsets so that these required displacements do not collide.

\begin{lemma}[No displacement collisions]
\label{lem:offsets}
The $r^2$ members of $\mathcal D$ are positive, distinct, and smaller than $(q-1)/2$. The positions $d_i,-d_i$ are all distinct modulo $q-1$.
\end{lemma}
\begin{proof}
If $\eta_i+\eta_j=\eta_u+\eta_v$, reduction modulo $2r+1$ gives $i+j=u+v$, since both sums lie between $2$ and $2r$. Comparing the remaining terms gives $i^2+j^2=u^2+v^2$, hence $ij=uv$ and $\{i,j\}=\{u,v\}$ as multisets. This includes diagonal sums.

Likewise, if $\eta_j-\eta_i=\eta_v-\eta_u$ with $i<j$ and $u<v$, reduction modulo $2r+1$ gives $j-i=v-u$. Comparing the quadratic terms gives $(j-i)(j+i)=(v-u)(v+u)$; the common nonzero difference implies $i+j=u+v$, so $(i,j)=(u,v)$.

Now $0<d_j-d_i<2\eta_r$, whereas every sum $d_i+d_j$, including $2d_i$, exceeds $4B_0$. Thus the difference and sum ranges are disjoint, and the two uniqueness arguments show that all $r+2\binom r2=r^2$ displacements are distinct. They are at most $2d_r$. To check that they are smaller than $(q-1)/2$, expand
\[
 d_r=44r^3+22r^2+22r+21
\]
and use $k=6(r+3)$ to obtain
\[
 k^3-4d_r=40r^3+1856r^2+5744r+5748>0.
\]
Since the lower endpoint in~\eqref{eq:parameters} is an integer,
\[
 4d_r<k^3<(100k)^{12}\le q-1.
\]
Hence none wraps around or equals the negative of a checked displacement. Finally $0<d_i+d_j<q-1$, so $d_i\not\equiv-d_j$.
\end{proof}

\begin{lemma}[Bounds supplied by the parameters]
\label{lem:parameterbounds}
For the parameters in~\eqref{eq:parameters}, $k\ge24$, $h+1<2k<q$, $q>h^2$, and
\[
 h-r-3=\frac{4k}{3}\ge k+2.
\]
Moreover,
\[
 q>16\bigl(h^2+2h(2r+2)\bigr),\qquad
 q>64r^2h(r+h).
\]
For every integer $s'$ with $h-r-3\le s'\le h$,
\[
 (4k)^{s'}q^{-(s'-k)/2}
 <(100k)^{-k/2}<\frac18.
\]
\end{lemma}
\begin{proof}
The identities $k=6(r+3)$ and $h=3k/2$, with $r\ge1$, give $k\ge24$, $h+1<2k$, and $h-r-3=4k/3\ge k+2$. The chosen prime range gives the stronger bounds
\[
 q>(100k)^{12}>
 \max\left\{52k^2,\frac{40}{9}k^4\right\}.
\]
In particular, $q>2k$ and $q>h^2=9k^2/4$. Since $r<k/6$ and $2r+2<k/3$, the two collision bounds follow explicitly:
\begin{align*}
 16\bigl(h^2+2h(2r+2)\bigr)
 &<16\left(\frac{9k^2}{4}+3k\cdot\frac{k}{3}\right)
 =52k^2<q,\\
 64r^2h(r+h)
 &<64\left(\frac{k}{6}\right)^2\frac{3k}{2}
       \left(\frac{k}{6}+\frac{3k}{2}\right)
 =\frac{40}{9}k^4<q.
\end{align*}
For the counting error, use the full lower bound $q>(100k)^{12}$, together with $s'\le3k/2$ and $s'-k\ge k/3$. Therefore
\[
 (4k)^{s'}q^{-(s'-k)/2}
 <(100k)^{3k/2-12k/6}
 =(100k)^{-k/2}<\frac18.
\]
The last inequality uses $k\ge24$. These estimates will control the counting error and the two collision union bounds.
\end{proof}

\section{The cyclic parity-anchor lattice}
\label{sec:anchor}
Coordinate $a\in\Z/(q-1)\Z$ carries the field label $\alpha^a$. For $y\in\Z^{q-1}$ define
\[
 u(y)=\sum_{a\text{ even}}y_a,\qquad
 v(y)=\sum_{a\text{ odd}}y_a,\qquad
 M_j(y)=\sum_{a=0}^{q-2}y_a\alpha^{aj}\pmod q.
\]
The parity sums $u,v$ are integers; the moments $M_j$ lie in $\F_q$. Define
\begin{equation}
\label{eq:ideal}
 \mathcal I=\left\{y\in\Z^{q-1}:\ 
 hu(y)-v(y)\equiv0\pmod{h^2-1},\quad
 M_j(y)=0\ (1\le j<k)\right\}.
\end{equation}
The parity condition has the equivalent forms
\begin{equation}
\label{eq:paritycong}
 \binom uv\in\begin{pmatrix}h&1\\1&h\end{pmatrix}\Z^2
 \iff hu-v\equiv0\pmod{h^2-1}
 \iff hv-u\equiv0\pmod{h^2-1}.
\end{equation}
Indeed,
\[
 \begin{pmatrix}h&1\\1&h\end{pmatrix}^{-1}\binom uv
 =\frac1{h^2-1}\binom{hu-v}{hv-u},
\]
and $hv-u=(h^2-1)u-h(hu-v)$, with the symmetric identity obtained by interchanging $u,v$. Thus divisibility of either numerator implies divisibility of both.

\begin{lemma}[A single modular kernel; ideal property and full rank]
\label{lem:ideal}\label{lem:baseann}
One can compute $P\in R_{q-1}$ in deterministic polynomial time such that
\begin{equation}
\label{eq:baseann}
 \boxed{\mathcal I=\{y\in R_{q-1}:Py\equiv0\pmod{q(h^2-1)R_{q-1}}\}.}
\end{equation}
Consequently $\mathcal I$ is a full-rank ideal of $R_{q-1}$,
$q(h^2-1)R_{q-1}\subseteq\mathcal I$, and $S\mathcal I=\mathcal I$.
\end{lemma}
\begin{proof}
We convert the parity and moment conditions into one multiplication congruence. The original description~\eqref{eq:ideal} remains the useful one for the shortest-vector argument.

\textbf{Cancellation over a coefficient ring.}
If $\mathcal A$ is a commutative ring and $f=bp$ in $\mathcal A[X]$, with $p$ monic, then
\begin{equation}
\label{eq:monicann}
 \{y\in\mathcal A[X]/(f):py=0\}=(b).
\end{equation}
Indeed, $pb=f$ gives one inclusion. Conversely, a lift of $py=0$ satisfies $pY=fZ=pbZ$, so $p(Y-bZ)=0$ in $\mathcal A[X]$. Multiplication by a monic polynomial is injective: the highest nonzero coefficient of a nonzero polynomial remains nonzero in its product with $p$. Thus $Y=bZ$. This argument also works when the coefficient ring has zero divisors.

\textbf{The parity condition as evaluation at $-h$.}
Modulo $h^2-1$, we have $h^2\equiv1$, and hence
\[
 (-h)^{2t}\equiv1,\qquad (-h)^{2t+1}\equiv-h.
\]
Writing $y(X)=\sum_{a=0}^{q-2}y_aX^a$ therefore gives
\[
 y(-h)\equiv\sum_{a\text{ even}}y_a-h\sum_{a\text{ odd}}y_a
 =u(y)-hv(y)\pmod{h^2-1}.
\]
This evaluation is well defined on the quotient ring: $q-1$ is even, so $(-h)^{q-1}\equiv1\pmod{h^2-1}$. Also $h$ is a unit modulo $h^2-1$, with inverse $h$. Multiplying $hu-v$ by $h$ shows that the parity condition is equivalent to $y(-h)=0$ modulo $h^2-1$.

Division by the monic polynomial $X+h$ identifies its remainder with this evaluation. Define the complementary monic factor
\[
 P_{\rm par}(X)=\frac{X^{q-1}-1}{X+h}
 \quad\text{in }(\Z/(h^2-1)\Z)[X].
\]
This is a factorization over the displayed coefficient ring, not over $\Z[X]$. By~\eqref{eq:monicann}, the parity condition is equivalent to
\[
 P_{\rm par}y=0\quad\text{in }R_{q-1}/(h^2-1)R_{q-1}.
\]

\textbf{The moment conditions.}
The equations are $y(\alpha^j)=0$ for $1\le j<k$. Since the $\alpha^j$ are distinct roots of $X^{q-1}-1$, define
\[
 P_{\rm mom}(X)=
 \frac{X^{q-1}-1}{\displaystyle\prod_{j=1}^{k-1}(X-\alpha^j)}
 \quad\text{in }\F_q[X].
\]
Vanishing at these roots is equivalent to divisibility by the denominator. Applying~\eqref{eq:monicann} again, the moment equations are equivalent to $P_{\rm mom}y=0$ in $R_{q-1}/qR_{q-1}$.

\textbf{Combine the two conditions.}
Since $q>h^2$, the moduli $q$ and $h^2-1$ are coprime. Coefficientwise Chinese remaindering gives $P$ of degree less than $q-1$ with
\[
 P\equiv P_{\rm par}\pmod{h^2-1},\qquad
 P\equiv P_{\rm mom}\pmod q.
\]
The two vanishing conditions are therefore exactly~\eqref{eq:baseann}. Monic polynomial division and the integer Euclidean algorithm compute $P$ in polynomial time; its coefficients may be taken in $\{0,\ldots,q(h^2-1)-1\}$. No factorization of $h^2-1$ is needed.

Finally, multiplication by $P$ is $R_{q-1}$-linear. Its kernel modulo $q(h^2-1)$ is thus an ideal, and it contains $q(h^2-1)R_{q-1}$, which already has full integer rank. Ideal closure gives $S\mathcal I\subseteq\mathcal I$; the finite order of $S$ gives equality.
\end{proof}

The two descriptions have different roles. The parity and moment equations expose the geometry of shortest vectors. The single kernel~\eqref{eq:baseann} provides a direct basis construction in Section~\ref{sec:algorithm} and the starting point for the NTRU application in Section~\ref{sec:ntru}.

\subsection{A moment-lattice lower bound}
The following Reed--Solomon lattice estimate is the one used in the construction of Bennett and Peikert~\cite{BP}. We include its elementary proof, so no lattice-hardness theorem is needed for this step.
\begin{lemma}[Reed--Solomon minimum bound]
\label{lem:RS}
Let $q$ be prime, $2k\le q$, and $x\in\Z^{\F_q}$ satisfy
\[
 \sum_{a\in\F_q}x_a a^j\equiv0\pmod q\qquad(0\le j<k),
\]
where $a^0=1$, including at $a=0$. If $x\ne0$, then $\norm{x}_2^2\ge2k$.
\end{lemma}
\begin{proof}
Suppose $\norm{x}_2^2<2k$. Integrality gives $\norm{x}_1\le\norm{x}_2^2<q$, so the zeroth congruence implies the integer equality $\sum_a x_a=0$.

Form a positive multiset with $x_a$ copies of $a$ when $x_a>0$, and a negative multiset with $-x_a$ copies when $x_a<0$. Both have size $d=\norm{x}_1/2\le k-1$. Their first $d$ power sums agree. Since $1,\ldots,d$ are invertible in $\F_q$, Newton's identities give the same elementary symmetric functions and hence the same monic root polynomial. The multisets are equal, including multiplicities. Their supports are disjoint by construction, so both must be empty. This forces $x=0$, a contradiction.
\end{proof}

\subsection{The minimum shell and its anchor}
\begin{proposition}[Relative anchor]
\label{prop:shell}
Every nonzero $y\in\mathcal I$ has $\norm{y}_2^2\ge h+1$. Equality holds exactly at signed binary vectors with vanishing moments and parity weights $(h,1)$ or $(1,h)$. Such vectors exist by Proposition~\ref{prop:clean}; hence $\lambda_1(\mathcal I)^2=h+1$.
\end{proposition}
\begin{proof}
First suppose $u(y)=v(y)=0$. The coordinates of $y$ are indexed by exponents, whereas Lemma~\ref{lem:RS} uses field labels. Define $x\in\Z^{\F_q}$ by
\[
 x_0=0,\qquad x_{\alpha^a}=y_a\quad(0\le a\le q-2).
\]
Thus the appended coordinate has field label $0$; the original coordinate at exponent $0$ has field label $1$. The zeroth moment is
\[
 \sum_{b\in\F_q}x_b=u(y)+v(y)=0,
\]
and for $1\le j<k$,
\[
 \sum_{b\in\F_q}x_b b^j
 =\sum_{a=0}^{q-2}y_a\alpha^{aj}=0\quad\text{in }\F_q.
\]
The labeling is a bijection and the appended entry is zero, so $x\ne0$ and $\norm{x}_2^2=\norm{y}_2^2$. Lemma~\ref{lem:RS} therefore gives $\norm{y}_2^2\ge2k>h+1$.

For $(u,v)\ne(0,0)$, the matrix condition in~\eqref{eq:paritycong} gives integers $a,b$ with $u=ha+b$ and $v=a+hb$. Use the identity
\[
 |u|+|v|=\max\{|u+v|,|u-v|\}
 =\max\{(h+1)|a+b|,(h-1)|a-b|\}.
\]
The first equality follows by considering whether $u,v$ have the same sign. If $a+b\ne0$, the first term is at least $h+1$. If $a+b=0$, then $a=-b\ne0$, so the second term is at least $2(h-1)>h+1$. Therefore
\[
 \norm{y}_2^2\ge\norm{y}_1\ge|u|+|v|\ge h+1.
\]

\textbf{Equality cases.} If $|u|+|v|=h+1$, then
\[
 |a+b|\le1,\qquad |a-b|\le\frac{h+1}{h-1}<2.
\]
The integers $a+b$ and $a-b$ have the same parity. The even possibility $(0,0)$ would give $u=v=0$, so both are $\pm1$. Hence
\[
 (a,b)\in\{(1,0),(0,1),(-1,0),(0,-1)\},
\]
and consequently
\[
 (u,v)\in\{(h,1),(1,h),(-h,-1),(-1,-h)\}.
\]
If $\norm{y}_2^2=h+1$, equality also holds in $\sum y_a^2\ge\sum|y_a|$, forcing $y_a\in\{0,\pm1\}$. Equality in the triangle inequality within each parity class forces all its nonzero entries to have the sign of that parity sum. The four displayed pairs have matching signs, so all nonzero entries share one sign. Thus $y$ is binary or the negative of a binary vector, with the asserted weights. Conversely, every such vector with vanishing moments belongs to $\mathcal I$ and has squared norm $h+1$.
\end{proof}

We call a coordinate \emph{occupied} when it is nonzero. In a minimum vector, each occupied coordinate has absolute value one and all have the same sign. The parity sum of absolute value $1$ therefore comes from exactly one occupied position: the \emph{anchor}. The other parity class has $h$ occupied positions.

A sign change and a cyclic rotation make the vector binary and put its anchor at $0$. It is then the only occupied even position. All other occupied positions are odd. Figure~\ref{fig:anchor} displays this normalized pattern; the odd entries need not be consecutive.

\begin{figure}[H]
\centering
\renewcommand{\arraystretch}{1.35}
\begin{tabular}{@{}lccccc@{\qquad}l@{}}
\toprule
Even positions & $0$ & $2$ & $4$ & $\cdots$ & $q-3$ & Occupied count\\
Entries & $\boxed{1}$ & $0$ & $0$ & $\cdots$ & $0$ & exactly $1$\\
\midrule
Odd positions & $1$ & $3$ & $5$ & $\cdots$ & $q-2$ & Occupied count\\
Entries & $\varepsilon_1$ & $\varepsilon_3$ & $\varepsilon_5$ & $\cdots$ & $\varepsilon_{q-2}$ & exactly $h$\\
\bottomrule
\end{tabular}
\[
 \varepsilon_a\in\{0,1\},\qquad
 \sum_{a\text{ odd}}\varepsilon_a=h.
\]
\caption{A normalized minimum vector. The boxed entry is the unique anchor. The moment equations further restrict which $h$ odd positions may be occupied; the diagram describes the support pattern, not an arbitrary choice of those positions.}
\label{fig:anchor}
\end{figure}

\par\noindent\begin{minipage}{\linewidth}
\begin{definition*}[Nonsquare labels]
With the fixed primitive element $\alpha\in\F_q^\times$, define
\[
 \mathcal Q:=\alpha(\F_q^\times)^2
 =\{\alpha^a:a\in\Z/(q-1)\Z\text{ has odd parity}\}.
\]
It is the set of nonzero nonsquares in $\F_q$, and $|\mathcal Q|=(q-1)/2$.
\end{definition*}
\end{minipage}\par

The map $a\mapsto\alpha^a$ is a bijection from coordinate positions to $\F_q^\times$; on odd positions it is a bijection onto $\mathcal Q$. Thus the dense support is equivalently an $h$-element subset of $\mathcal Q$. The anchor has label $\alpha^0=1$, so the $h$ other labels must have power sums $-1$ for $1\le j<k$. All normalizations are by cyclic shifts and an overall sign change.

\Needspace{8\baselineskip}
\section{Uniform counting in a quadratic-character coset}
\label{sec:count}
Recall that $\mathcal Q=\alpha(\F_q^\times)^2$ is the set of $(q-1)/2$ labels of odd positions. To count $h$-element supports with prescribed entries and moments, write each label as $\alpha x^2$, count the resulting tuples, and exclude zeros, repetitions, and forbidden labels.

\subsection{The counting inputs}
The two estimates below have different roles: the first gives the error exponent over every extension field; the second gives a uniform bound for its coefficient.

\textbf{Deligne--Hooley--Katz.} Let $V\subseteq\PP^n_{\F_q}$ be a projective complete-intersection scheme of dimension $d\ge1$, and let $\sigma=\dim\Sing(V_{\overline{\F}_q})$, with $\sigma=-1$ when $V$ is smooth. The singular locus is that of the scheme after base change, not of its reduction. No separate irreducibility or reducedness hypothesis is imposed. For every $e\ge1$,
\begin{equation}
\label{eq:DHK}
 \left|\#V(\F_{q^e})-\sum_{i=0}^d q^{ei}\right|
 \le C_V q^{e(d+1+\sigma)/2},
\end{equation}
where $C_V$ is independent of $e$. The smooth case follows from Deligne~\cite{Deligne}; the singular estimate is due to Hooley with Katz's appendix~\cite{Hooley}. The all-extension formulation is also recorded in our earlier work~\cite[Proposition 5.1]{Wan}. We do not assume that $C_V$ is uniform over our fibers.

In the applications below, $\sigma\in\{-1,0\}$ and $\sigma\le d-2$. This bound makes $V_{\overline{\F}_q}$ regular in codimension one. Complete intersections are Cohen--Macaulay, hence satisfy $S_2$; Serre's criterion therefore makes $V_{\overline{\F}_q}$ normal~\cite[Tags 00S8 and 031S]{Stacks}. A positive-dimensional projective complete intersection is geometrically connected, so normality implies geometric integrality; see also~\cite[discussion after Proposition 5.1]{Wan}. Thus nonreduced schemes do not occur in this small-singular-locus range. When $\sigma\ge d-1$, the exponent in~\eqref{eq:DHK} is at least $d$, and a dimension-based point bound proves the estimate. Point counts ignore nilpotents, but the singular-locus hypothesis must still refer to the original scheme.

\textbf{Wan--Zhang.} If $V\subseteq\Aff^n$ is cut out by $c\ge1$ equations of degree at most $D$ and has dimension $n-c$, its total compactly supported Betti number satisfies
\begin{equation}
\label{eq:WZ}
 B_c(V,\ell):=\sum_i\dim H_c^i(V_{\overline{\F}_q},\Q_\ell)
 \le\binom{n-1}{c-1}(D+1)^n,
 \qquad\ell\ne\operatorname{char}(\F_q).
\end{equation}
This is Wan--Zhang~\cite[Theorem 1.2.2 and Corollary 4.2.4(ii)]{WZ} in the cited version. It allows singular affine complete intersections.

We also use rationality and the cohomological trace formula for
\[
 Z(V,t)=\exp\left(\sum_{e\ge1}\#V(\F_{q^e})\frac{t^e}{e}\right).
\]
After cancellation, the sum of the numerator and denominator degrees of this rational function is at most $B_c(V,\ell)$. This follows by expressing the zeta function as the alternating product of Frobenius characteristic polynomials on compactly supported cohomology; see also~\cite[Section 1.1]{WZ}. The passage from all-extension counts to a uniform constant below follows the method of~\cite[proof of Proposition 5.2]{Wan}.

\Needspace{15\baselineskip}
\subsection{A uniform weighted-fiber estimate}
\begin{lemma}[Uniform fiber estimate]
\label{lem:fiber}
Let $q$ be prime, $k\ge2$, $s\ge k+2$, and $q>2(k-1)$. Let $c_1,\ldots,c_s$ be positive integers with $\sum_i c_i<q$. For arbitrary $b_1,\ldots,b_{k-1}\in\F_q$, define
\[
 X=\left\{x\in\Aff^s:\ \sum_{i=1}^s c_i x_i^{2j}=b_j,
 \quad 1\le j<k\right\}.
\]
For every $e\ge1$,
\begin{equation}
\label{eq:fiber}
 \left|\#X(\F_{q^e})-q^{e(s-k+1)}\right|
 \le\mathfrak B_{s,k}\,q^{e(s-k+2)/2},\qquad
 \mathfrak B_{s,k}=\binom{s-1}{k-2}(2k-1)^s.
\end{equation}
The coefficient is uniform in the syndrome, the permitted weights, $q$, and $e$, and satisfies $\mathfrak B_{s,k}<(4k)^s$.
\end{lemma}
\begin{proof}
We separate the geometry, the error exponent, and the uniform coefficient.

\stepheading{1. Geometry of the fiber and its section at infinity.}
Work over $\overline{\F}_q$. Let
\[
 Z_\infty\subseteq\PP^{s-1}:\quad
 \sum_i c_i x_i^{2j}=0\quad(1\le j<k).
\]
After removing nonzero row factors $2j$ and column factors $c_i x_i$, the nonzero columns of its Jacobian have the form
\[
 (1,x_i^2,x_i^4,\ldots,x_i^{2k-4})^{\mathsf T}.
\]
Thus $k-1$ distinct nonzero square values give a nonsingular Vandermonde minor.

Suppose instead that a point of $Z_\infty$ has only $t\le k-2$ distinct nonzero square values $a_1,\ldots,a_t$. Here $t\ge1$ because the point is projective. Group the weights belonging to each square value, obtaining positive integers $w_1,\ldots,w_t$ with $w_\nu<q$. The first $t$ equations imply
\[
 \begin{pmatrix}
 a_1&\cdots&a_t\\
 a_1^2&\cdots&a_t^2\\
 \vdots&&\vdots\\
 a_1^t&\cdots&a_t^t
 \end{pmatrix}
 \begin{pmatrix}w_1\\w_2\\\vdots\\w_t\end{pmatrix}=0.
\]
Its determinant is $\bigl(\prod_\nu a_\nu\bigr)\prod_{\mu<\nu}(a_\nu-a_\mu)\ne0$, forcing every $w_\nu$ to vanish in the field. This contradicts $0<w_\nu<q$. For $k=2$, any nonzero coordinate already gives the required Jacobian rank.

The projective dimension theorem gives a nonempty intersection, and the Jacobian calculation gives dimension exactly $s-k$. Hence $Z_\infty$ is a smooth projective complete intersection. Its dimension is at least two. A positive-dimensional projective complete intersection is geometrically connected, so smoothness makes $Z_\infty$ geometrically integral.

Now homogenize the affine equations:
\[
 Y\subseteq\PP^s:\quad F_j=\sum_i c_i x_i^{2j}-b_jt^{2j}=0,
 \qquad 1\le j<k.
\]
Let $H_\infty=\{t=0\}\subset\PP^s$, so $Y\cap H_\infty=Z_\infty$. Put
\[
 d=s-k+1\ge3.
\]
We first prove directly that every irreducible component of $Y$ has dimension $d$.

Let $Y_0$ be an irreducible component, with its reduced structure. Since $Y$ is cut out by $k-1$ equations in $\PP^s$, the dimension theorem gives $\dim Y_0\ge s-(k-1)=d$. The component cannot lie in $H_\infty$: otherwise it would be contained in $Z_\infty$, whose dimension is only $d-1$. Because $Y_0$ is projective of positive dimension and is not contained in the hyperplane, its hyperplane section is nonempty and
\[
 \dim(Y_0\cap H_\infty)=\dim Y_0-1.
\]
But $Y_0\cap H_\infty\subseteq Z_\infty$, so $\dim Y_0-1\le d-1$. Combining the two inequalities yields
\[
 \boxed{\dim Y_0=d\quad\text{for every irreducible component }Y_0.}
\]
The projective dimension theorem used here is recalled in~\cite[Tag 0B2N]{Stacks}. The $k-1$ equations therefore have the expected codimension. In the regular ambient space they define a complete intersection; equivalently, their local equations form a regular sequence. Its local rings are Cohen--Macaulay~\cite[Tag 00S8]{Stacks}.

We also have $\dim(Y_0\cap H_\infty)=d-1=\dim Z_\infty$. Since $Z_\infty$ is irreducible, the closed subset $Y_0\cap H_\infty$ has the same underlying closed set as $Z_\infty$. Thus every irreducible component of $Y$ contains $Z_\infty$. The Jacobian in the $x_i$ variables has full rank there, so $Y$ is smooth along $Z_\infty$. Two distinct components would meet there, contradicting smoothness. Hence $Y$ has a single irreducible component.

The smooth points along $Z_\infty$ also show that $Y$ is generically reduced. Its Cohen--Macaulay local rings satisfy Serre's condition $S_1$, and generic reducedness gives $R_0$. The Noetherian reducedness criterion $R_0+S_1$~\cite[Tag 031R]{Stacks} makes $Y$ reduced. Thus $Y$ and its nonempty affine open $X=Y\setminus H_\infty$ are geometrically integral. Finally, the singular locus of $Y$ is closed and disjoint from $H_\infty$. A positive-dimensional projective closed set meets every hyperplane, so $\dim\Sing(Y)\le0$.

\stepheading{2. An error exponent over every extension.}
Apply~\eqref{eq:DHK} to $Y$ and its smooth section $Z_\infty$. Subtracting their counts and their projective-space main terms gives, for every $e\ge1$,
\begin{equation}
\label{eq:extensionerror}
 \left|\#X(\F_{q^e})-q^{ed}\right|
 \le(C_Y+C_Z)q^{e(d+1)/2}.
\end{equation}
The constants may depend on the fiber, but not on $e$. Only this exponential rate is needed in the next step.

\stepheading{3. A uniform coefficient from the reduced zeta function.}
The rational function
\[
 R(t)=(1-q^dt)Z(X,t),\qquad R(0)=1,
\]
has logarithmic derivative
\[
 \frac{R'(t)}{R(t)}=\sum_{e\ge1}\bigl(\#X(\F_{q^e})-q^{ed}\bigr)t^{e-1}.
\]
By~\eqref{eq:extensionerror}, this series is holomorphic for $|t|<q^{-(d+1)/2}$. Cancel common factors and write
\[
 R(t)=\prod_\lambda(1-\lambda t)^{m_\lambda},
 \qquad m_\lambda\in\Z\setminus\{0\},
\]
with distinct nonzero $\lambda$. Its logarithmic derivative is
\[
 -\sum_\lambda\frac{m_\lambda\lambda}{1-\lambda t}.
\]
A factor with $|\lambda|>q^{(d+1)/2}$ would give an uncancelled pole inside the disk of holomorphy. Therefore every surviving reciprocal factor satisfies $|\lambda|\le q^{(d+1)/2}$.

Since $d>(d+1)/2$, $R$ has neither a zero nor a pole at $q^{-d}$. Hence multiplication by $1-q^dt$ removed exactly one simple pole of $Z(X,t)$, and the trace formula gives
\[
 \sum_\lambda|m_\lambda|\le B_c(X,\ell)-1.
\]
Comparing logarithmic coefficients now yields
\[
 \left|\#X(\F_{q^e})-q^{ed}\right|=\left|-\sum_\lambda m_\lambda\lambda^e\right|
 \le\bigl(B_c(X,\ell)-1\bigr)q^{e(d+1)/2}.
\]
Finally, $X\subseteq\Aff^s$ is cut out by $k-1$ equations of maximum degree $2(k-1)$ and has dimension $s-(k-1)$. Equation~\eqref{eq:WZ} gives
\[
 B_c(X,\ell)\le\binom{s-1}{k-2}(2k-1)^s=\mathfrak B_{s,k}.
\]
Dropping the harmless subtraction of $1$ proves~\eqref{eq:fiber}. The inequality $\binom{s-1}{k-2}\le2^{s-1}$ gives $\mathfrak B_{s,k}<(4k)^s$.
\end{proof}

\subsection{Prescribed subsets and inclusion probabilities}
There are $(q-1)/2$ available labels, one for each odd coordinate. A binary dense support of size $h$ is an $h$-element subset of $\mathcal Q$.

Let $P_+,P_-\subseteq\mathcal Q$ be disjoint. Labels in $P_+$ are prescribed to lie in the support; labels in $P_-$ are prescribed not to lie in it. All other labels, in $\mathcal Q\setminus(P_+\cup P_-)$, are \emph{free}: the completion may include some and omit the rest, subject to its size and moment constraints. Put $t=|P_+|$ and assume
\[
 t\le r+2,\qquad |P_+\cup P_-|\le2r+2.
\]
For $\mathbf b\in\F_q^{k-1}$, let $\mathcal C(P_+,P_-;\mathbf b)$ be the family of subsets $T\subseteq\mathcal Q$ satisfying
\[
 |T|=h,\qquad P_+\subseteq T,\qquad T\cap P_-=\varnothing,
 \qquad \sum_{a\in T}a^j=b_j\quad(1\le j<k).
\]
Thus exactly $h-t$ further occupied labels must be chosen from the free pool.

We will also need the probability that one or two specified free labels are both included in a uniformly chosen completion. Section~\ref{sec:clean} uses the one-label bound for an unwanted pair with one fixed occupied endpoint, and the two-label bound when both endpoints are free. These probabilities are ratios of counts with one or two extra occupied prescriptions. The limits above must hold \emph{after} those extra prescriptions as well.

\begin{lemma}[Subset counts and inclusion probabilities]
\label{lem:subsets}
Under~\eqref{eq:parameters},
\begin{equation}
\label{eq:subsets}
 \frac12\frac{q^{h-t-k+1}}{2^{h-t}(h-t)!}
 \le|\mathcal C(P_+,P_-;\mathbf b)|
 \le2\frac{q^{h-t-k+1}}{2^{h-t}(h-t)!}.
\end{equation}
In a uniformly chosen member, the probability that one or two specified distinct free labels are occupied is at most, respectively,
\begin{equation}
\label{eq:conditional}
 \frac{8h}{q}\qquad\text{and}\qquad\frac{16h^2}{q^2},
\end{equation}
provided the enlarged prescribed sets satisfy the same size limits.
\end{lemma}
\begin{proof}
\stepheading{Ordered tuples.}
Put $s=h-t$ and write each residual label as $\alpha x_i^2$. The equations become
\[
 \sum_{i=1}^s x_i^{2j}
 =\alpha^{-j}\left(b_j-\sum_{a\in P_+}a^j\right),\qquad 1\le j<k.
\]
Lemma~\ref{lem:fiber} counts all tuples. We must exclude zero coordinates, repeated squares, and labels in $P_+\cup P_-$. Each forbidden event reduces to a fiber covered by the same lemma:
\begin{center}
\small
\begin{tabular}{@{}p{0.30\textwidth}p{0.61\textwidth}@{}}
\toprule
Forbidden event & Corresponding substitution\\
\midrule
$x_i=0$ & Delete the variable and its weight.\\[1mm]
$x_i=x_j$ or $x_i=-x_j$ & Merge the variables, adding their weights; the even powers give the same equation for either sign.\\[1mm]
$\alpha x_i^2=a\in P_+\cup P_-$ & For each of its two square roots, fix $x_i$, delete it, and subtract its contribution from the syndrome.\\
\bottomrule
\end{tabular}
\end{center}
All resulting weights remain positive, with total at most $h<q$.

\stepheading{Uniform error control.}
Every prescribed family counted here---including a numerator family with one or two additional occupied labels---has at most $r+2$ occupied prescriptions in total. Its unrestricted tuple therefore has at least $h-r-2$ variables. Each individual forbidden substitution deletes one variable or merges two into one, decreasing that count by one. Thus every fiber used in the proof has
\[
 h-r-3\le s'\le h.
\]
In particular, adding two labels is permitted only when the original $t\le r$ (and the total-prescription limit also holds), not when $t=r+2$. We bound single forbidden events for a union bound; intersections of several such events need not be counted.
Recall that
\[
 h=\frac{3k}{2},\qquad (100k)^{12}<q<2(100k)^{12}.
\]
Here $r+3=k/6$, so $s'\ge h-r-3=4k/3$. In particular, $s'-k\ge k/3$ and $s'\ge k+2$, as recorded in Lemma~\ref{lem:parameterbounds}. To see the error exponent explicitly, use
\[
 (4k)^{s'}\le(100k)^{3k/2},
\]
and
\[
 q^{-(s'-k)/2}
 <(100k)^{-12(s'-k)/2}
 \le(100k)^{-2k}.
\]
The exponent $2k$ is $12(k/3)/2$. Therefore the relative error in Lemma~\ref{lem:fiber} satisfies
\[
 \mathfrak B_{s',k}q^{-(s'-k)/2}
 <(100k)^{3k/2-2k}
 =(100k)^{-k/2}<\frac18.
\]
Thus the unrestricted count lies between $(7/8)q^{s-k+1}$ and $(9/8)q^{s-k+1}$, while each forbidden substitution contributes at most $(9/8)q^{s-k}$.

\stepheading{Remove forbidden tuples and forget their order.}
Since $s=h-|P_+|\le h$ and $|P_+\cup P_-|\le2r+2$, the number of forbidden substitutions is at most
\[
 \begin{aligned}
 s+2\binom s2+2s|P_+\cup P_-|
 &=s^2+2s|P_+\cup P_-|\\
 &\le h^2+2h(2r+2).
 \end{aligned}
\]
Each substitution contributes at most $\frac98q^{s-k}$ tuples. Lemma~\ref{lem:parameterbounds} gives $q>16\bigl(h^2+2h(2r+2)\bigr)$, so the union of the forbidden events contains at most
\[
 \begin{aligned}
 \frac98\bigl(h^2+2h(2r+2)\bigr)q^{s-k}
 &<\frac9{128}q^{s-k+1}\\
 &<\frac18q^{s-k+1}
 \end{aligned}
\]
tuples. Subtracting this from the unrestricted lower bound $\frac78q^{s-k+1}$ leaves more than $\frac34q^{s-k+1}$ admissible tuples. The unrestricted upper bound is $\frac98q^{s-k+1}$; in particular, the admissible count lies between $\frac12q^{s-k+1}$ and $2q^{s-k+1}$, as required.

Each residual subset has exactly $s!$ orderings and two independent square-root choices for every label, hence corresponds to exactly $2^ss!$ admissible tuples. Division proves~\eqref{eq:subsets}.

\stepheading{Inclusion probabilities.}
Let $a=1$ or $2$, and fix distinct labels $b_1',\ldots,b_a'$ in the free pool such that the enlarged prescription remains within the stated limits. For a uniformly chosen $T\in\mathcal C(P_+,P_-;\mathbf b)$,
\[
 \Pr[\{b_1',\ldots,b_a'\}\subseteq T]
 =\frac{|\mathcal C(P_+\cup\{b_1',\ldots,b_a'\},P_-;\mathbf b)|}
 {|\mathcal C(P_+,P_-;\mathbf b)|}.
\]
\Needspace{8\baselineskip}
The numerator has $s-a$ residual labels. Dividing its upper bound by the denominator's lower bound gives
\[
 4\,2^a\frac{s!}{(s-a)!}\,q^{-a},
\]
which is at most $8h/q$ for $a=1$ and $16h^2/q^2$ for $a=2$. No independence of inclusion events is assumed.
\end{proof}

\Needspace{17\baselineskip}
\section{Clean completions for Boolean assignments}
\label{sec:clean}
The checker comparison of Section~\ref{sec:roadmap} is automatically sound for every shortest vector. To obtain completeness, we need one shortest vector for a YES witness on which all correlation upper bounds are exact. This section proves that such clean completions always exist.

For a binary vector $y\in\Z^{q-1}$ and a displacement $D\in\Z/(q-1)\Z$, the cyclic correlation
\[
 C_y(D)=\sum_{a\bmod(q-1)}y_a y_{a+D}
\]
counts occupied directed pairs $(a,a+D)$: both entries of the pair equal $1$.

\par\noindent\begin{minipage}{\linewidth}
\begin{definition*}[Clean completion]
Fix $\xi\in\{0,1\}^r$. A \emph{clean completion} of $\xi$ is a binary vector $y\in\Z^{q-1}$ with
\[
 y_0=1,\qquad y_a=0\ \text{for every even }a\ne0,
 \qquad\sum_{a\text{ odd}}y_a=h,
\]
whose prescribed witness entries satisfy
\[
 y_{d_i}=\xi_i,\qquad y_{-d_i}=0\quad(1\le i\le r),
\]
and whose checked correlations satisfy
\begin{equation}
\label{eq:clean}
 \begin{gathered}
 C_y(2d_i)=0\quad(1\le i\le r),\\
 C_y(d_j-d_i)=\xi_i\xi_j,\qquad
 C_y(d_i+d_j)=0\quad(1\le i<j\le r).
 \end{gathered}
\end{equation}
\end{definition*}
\end{minipage}\par

Since the prescribed occupied pair $(d_i,d_j)$ is already present whenever $\xi_i=\xi_j=1$, the identities say that there are \emph{no additional occupied pairs} at any displacement in $\mathcal D$. Additional odd support is allowed, but not if it changes a checked correlation. The definition concerns the support and its correlations; the next proposition also imposes the moment equations.

Figure~\ref{fig:clean} displays the prescribed positions and the contributions that a clean completion must preserve.

\begin{figure}[H]
\centering
\begin{tikzpicture}[x=1cm,y=1cm,>=Stealth]
\foreach \x/\pos/\val in {-4/{-d_j}/{0},-2/{-d_i}/{0},0/{0}/{1},2/{d_i}/{\xi_i},4/{d_j}/{\xi_j}}{
  \node at (\x,0.43) {$\pos$};
  \node[draw,rounded corners=1pt,minimum width=0.65cm,minimum height=0.52cm] at (\x,0) {$\val$};
}
\node[font=\scriptsize,text=MidnightBlue] at (0,-0.52) {anchor};
\draw[->,thick,MidnightBlue] (2,0.78) -- (4,0.78)
 node[midway,above,font=\small] {$d_j-d_i$};
\draw[->,gray] (-2,-0.84) -- (2,-0.84)
 node[midway,below,font=\small] {$2d_i$};
\draw[->,gray] (-2,-1.65) -- (4,-1.65)
 node[midway,below,font=\small] {$d_i+d_j$};
\end{tikzpicture}
\medskip

\small
\begin{tabular}{@{}lll@{}}
\toprule
Checked displacement & Witness-position directed pairs & Total prescribed contribution\\
\midrule
$2d_i$ & $(-d_i,d_i)$ & $0$\\
$d_j-d_i$ & $(d_i,d_j),\ (-d_j,-d_i)$ & $\xi_i\xi_j$\\
$d_i+d_j$ & $(-d_i,d_j),\ (-d_j,d_i)$ & $0$\\
\bottomrule
\end{tabular}
\caption{Prescribed entries for $i<j$, drawn schematically and not to scale. The table lists all directed pairs among the witness positions that occur at the checked displacements. Only $(d_i,d_j)$ can be occupied. Free odd positions are omitted; the counting argument excludes any additional pair they might create at a checked displacement.}
\label{fig:clean}
\end{figure}
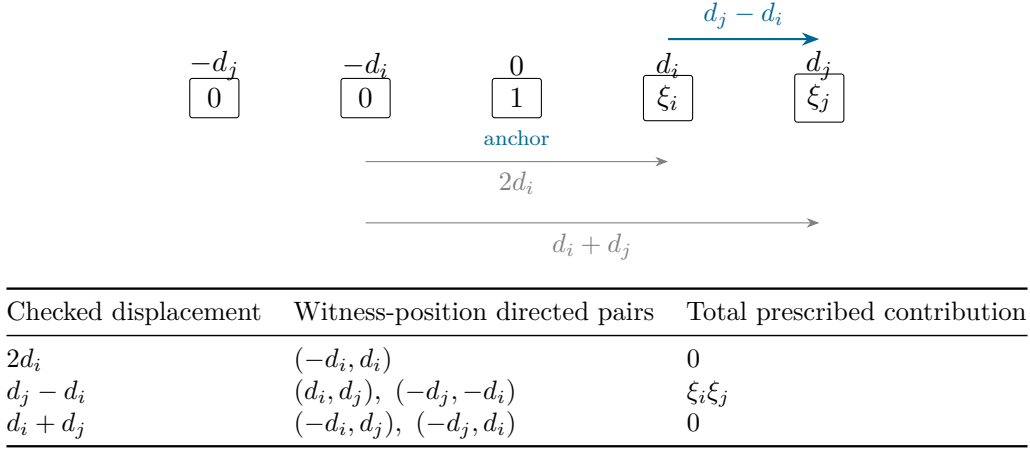

\begin{proposition}[Parity-restricted clean completion]
\label{prop:clean}
Every $\xi\in\{0,1\}^r$ has a clean completion $y\in\mathcal I$. In particular, $\norm{y}_2^2=h+1$, so $y$ lies in the minimum shell and has anchor at $0$.
\end{proposition}
\begin{proof}
Fix $\xi\in\{0,1\}^r$. We first construct the family to which the inclusion bounds will be applied.

\stepheading{1. Specify the family of supports.}
The positive witness position $d_i$ is prescribed to be occupied exactly when $\xi_i=1$, and every negative witness position $-d_i$ is prescribed to be unoccupied. In terms of field labels, set
\[
 \begin{aligned}
 P_+&=\{\alpha^{d_i}:\xi_i=1\},\\
 P_-&=\{\alpha^{-d_i}:1\le i\le r\}
       \cup\{\alpha^{d_i}:\xi_i=0\}.
 \end{aligned}
\]
All these labels lie in $\mathcal Q$, because their indices are odd. Lemma~\ref{lem:offsets} makes the $2r$ labels distinct, so $P_+\cap P_-=\varnothing$ and
\[
 |P_+|=\norm{\xi}_1\le r,\qquad |P_+\cup P_-|=2r.
\]
Let $\mathcal C_\xi$ denote the family $\mathcal C(P_+,P_-;(-1,\ldots,-1))$ of Section~\ref{sec:count}. Thus a member $T\subseteq\mathcal Q$ satisfies
\[
 |T|=h,\quad P_+\subseteq T,\quad T\cap P_-=\varnothing,
 \qquad \sum_{b\in T}b^j=-1\quad(1\le j<k).
\]
The prescription bounds of Lemma~\ref{lem:subsets} hold. Its positive lower bound implies that $\mathcal C_\xi$ is nonempty.

\stepheading{2. Each support already gives a minimum vector.}
For each $T\in\mathcal C_\xi$, define the binary vector $y(T)\in\Z^{q-1}$ by
\[
 y_a(T)=
 \begin{cases}
  1,&a=0,\\
  1,&a\text{ is odd and }\alpha^a\in T,\\
  0,&\text{otherwise}.
 \end{cases}
\]
The labeling map is a bijection, so there are exactly $h$ occupied odd positions. The prescriptions give $y_{d_i}=\xi_i$ and $y_{-d_i}=0$. Its parity sums are $(u(y),v(y))=(1,h)$, which satisfy the parity congruence. Its moments vanish because
\[
 M_j(y)=1+\sum_{b\in T}b^j=1-1=0\quad\text{in }\F_q
 \qquad(1\le j<k).
\]
Hence every $T\in\mathcal C_\xi$ gives $y(T)\in\mathcal I$ with
\[
 \norm{y(T)}_2^2=1+|T|=h+1.
\]
Proposition~\ref{prop:shell}'s lower bound makes all these vectors minimum vectors. What remains is to enforce the correlation identities~\eqref{eq:clean}.

\stepheading{3. Bound the probability of an unwanted pair.}
Choose $T$ uniformly from the finite nonempty family $\mathcal C_\xi$, and write $y=y(T)$. For $D\in\mathcal D$, call a directed pair $(a,a+D)$ \emph{unwanted} if it is occupied and is not the prescribed pair $(d_i,d_j)$ with $D=d_j-d_i$, $i<j$, and $\xi_i=\xi_j=1$. All indices here are modulo $q-1$.

\emph{Two prescribed occupied endpoints.}
The only prescribed occupied positions are $0$ and the $d_i$ with $\xi_i=1$. Every $D\in\mathcal D$ is even. The anchor at $0$ cannot form an occupied pair at such a displacement, because every other even position is zero. If $i<j$ and both $d_i$ and $d_j$ are prescribed occupied, then the directed pair $(d_i,d_j)$ occurs at the checked displacement $d_j-d_i$ and is intended. No unwanted pair is forced by two prescribed occupied endpoints.

Hence every remaining unwanted pair has at least one free odd endpoint. We partition these remaining pairs into two disjoint cases: exactly one endpoint is prescribed occupied and the other is free, or both endpoints are free. An endpoint prescribed to be zero can never contribute. The one- and two-label bounds in Lemma~\ref{lem:subsets} apply because the original family has at most $r$ occupied and exactly $2r$ total prescribed labels. Requiring one more free label to be occupied gives at most $r+1$ and $2r+1$, respectively; requiring two gives at most $r+2$ and $2r+2$. Both enlarged prescriptions meet the lemma's limits.

\Needspace{8\baselineskip}
\emph{One prescribed occupied endpoint and one free endpoint.}
Fix $D\in\mathcal D$. For a prescribed occupied odd position $d_i$, the only possible partners at displacement $D$ are $d_i+D$ and $d_i-D$, depending on which endpoint is first. There are at most $r$ such $d_i$, hence at most $2r$ candidate pairs of this type. A free partner has inclusion probability at most $8h/q$. The union bound therefore gives
\[
 \Pr[\text{an unwanted pair of this type at }D]
 \le 2r\,\frac{8h}{q}.
\]

\Needspace{7\baselineskip}
\emph{Two free endpoints.}
For each starting position $a$, there is just one directed pair $(a,a+D)$. There are therefore at most $q-1<q$ candidates with both endpoints free. The endpoints, and hence their field labels, are distinct because $D\ne0$. By the two-label bound their joint inclusion probability is at most $16h^2/q^2$. Thus
\[
 \Pr[\text{an unwanted pair with two free endpoints at }D]
 \le q\,\frac{16h^2}{q^2}.
\]
We use the joint inclusion estimate, not independence of the two events.

\stepheading{4. Choose one clean member of the family.}
There are $|\mathcal D|=r^2$ checked displacements. Applying the union bound to the two types above and to all these displacements gives
\begin{equation}
\label{eq:cleanprob}
 \begin{aligned}
 \Pr[\text{some unwanted pair}]
 &\le r^2\left(2r\frac{8h}{q}+q\frac{16h^2}{q^2}\right)\\
 &=\frac{16r^2h(r+h)}q<\frac14.
 \end{aligned}
\end{equation}
The last inequality follows from $q>64r^2h(r+h)$ in Lemma~\ref{lem:parameterbounds}. The probability of having no unwanted pair is therefore greater than $3/4$, in particular positive. At least one $T\in\mathcal C_\xi$ has this property.

For that support, the intended pair $(d_i,d_j)$ is present exactly when $\xi_i\xi_j=1$, and there are no other occupied pairs at $d_j-d_i$. At $2d_i$ and $d_i+d_j$ there are no intended occupied pairs and therefore no occupied pairs at all. Consequently
\[
 C_y(2d_i)=0,\qquad
 C_y(d_j-d_i)=\xi_i\xi_j,\qquad C_y(d_i+d_j)=0,
\]
which are precisely~\eqref{eq:clean}. The vector $y(T)$ already belongs to $\mathcal I$ and has squared norm $h+1$ by Step~2. It is the required clean completion.
\end{proof}

\section{A cyclic checker for the cover cost}
\label{sec:checker}
We now formalize the decoder and correlation comparisons from Section~\ref{sec:roadmap}. Recall that
\[
 C_y(d)=\sum_{a\bmod(q-1)}y_a y_{a+d}
\]
is defined for every integral $y\in\Z^{q-1}$. Reindexing gives $C_{S^t y}(d)=C_y(d)$, and $C_{-y}(d)=C_y(d)$, so correlations are unchanged by rotation or an overall sign reversal.

\begin{lemma}[Correlation inequalities]
\label{lem:corr}
Let $y$ be a minimum vector normalized so that it is binary and its anchor is at $0$. Then
\[
 C_y(d_i)=y_{d_i}+y_{-d_i},\qquad
 C_y(2d_i)\ge y_{d_i}y_{-d_i},
\]
and, for $i<j$,
\[
 C_y(d_j-d_i)+C_y(d_i+d_j)
 \ge (y_{d_i}+y_{-d_i})(y_{d_j}+y_{-d_j}).
\]
\end{lemma}
\begin{proof}
The anchor is the only occupied even position, and every $d_i$ is odd.

\textbf{Reading the two witness positions.} A nonzero term $y_a y_{a+d_i}$ has one even endpoint, which must be $0$. The only possibilities are $a=0$ and $a=-d_i$, contributing $y_{d_i}$ and $y_{-d_i}$. Thus the first identity is exact: additional support cannot change it.

\textbf{A product within one witness pair.} The term of $C_y(2d_i)$ starting at $-d_i$ is $y_{-d_i}y_{d_i}$. All other terms are nonnegative, proving the second inequality.

\textbf{Products between two witness pairs.} Expand
\[
 (y_{d_i}+y_{-d_i})(y_{d_j}+y_{-d_j}).
\]
Each of its four terms appears in one of the following correlations:
\begin{center}
\small
\begin{tabular}{@{}lll@{}}
\toprule
Product & Directed endpoints & Correlation\\
\midrule
$y_{d_i}y_{d_j}$ & $(d_i,d_j)$ & $C_y(d_j-d_i)$\\
$y_{-d_i}y_{-d_j}$ & $(-d_j,-d_i)$ & $C_y(d_j-d_i)$\\
$y_{-d_i}y_{d_j}$ & $(-d_i,d_j)$ & $C_y(d_i+d_j)$\\
$y_{d_i}y_{-d_j}$ & $(-d_j,d_i)$ & $C_y(d_i+d_j)$\\
\bottomrule
\end{tabular}
\end{center}
For example, $-d_j+(d_j-d_i)=-d_i$, explaining the orientation of the second row. Lemma~\ref{lem:offsets} makes the designated positions distinct, so the two terms in each correlation are different summands. All remaining terms are nonnegative, proving the last inequality.
\end{proof}

\subsection{From the repair polynomial to the checker}
For a minimum vector normalized as in Lemma~\ref{lem:corr}, define the decoded vector
\[
 z_i:=y_{d_i}+y_{-d_i}=C_y(d_i)\in\{0,1,2\},\qquad 1\le i\le r.
\]
This is precisely the decoder motivated in Section~\ref{sec:roadmap}. Since $z\ge0$, we may use the quadratic expansion of $\Psi_A$ recorded after Lemma~\ref{lem:repair}. The binary entries $y_{d_i},y_{-d_i}$ satisfy
\[
 z_i^2=z_i+2y_{d_i}y_{-d_i}.
\]
Substitution gives
\[
\boxed{\begin{aligned}
 \Psi_A(z)={}&m+\sum_i(1-|S_i|)z_i
    +2\sum_i|S_i|y_{d_i}y_{-d_i}\\
 &+2\sum_{i<j}|S_i\cap S_j|z_i z_j.
\end{aligned}}
\]
Lemma~\ref{lem:corr} supplies exactly the three replacements
\[
 \begin{aligned}
 z_i&=C_y(d_i),\\
 y_{d_i}y_{-d_i}&\le C_y(2d_i),\\
 z_i z_j&\le C_y(d_j-d_i)+C_y(d_i+d_j).
 \end{aligned}
\]
The first is exact, which is essential because $1-|S_i|$ may be negative. The other two are upper bounds multiplied by nonnegative coefficients. This gives the following anchor-free, shift-invariant checker.

\par\noindent\begin{minipage}{\linewidth}
\begin{definition*}[Cyclic cover checker]
For every $y\in\Z^{q-1}$, define
\begin{equation}
\label{eq:checker}
\begin{aligned}
 E_A(y)={}&m+\sum_{i=1}^r(1-|S_i|)C_y(d_i)
       +2\sum_{i=1}^r |S_i|\,C_y(2d_i)\\
       &+2\sum_{1\le i<j\le r}|S_i\cap S_j|
          \bigl(C_y(d_j-d_i)+C_y(d_i+d_j)\bigr).
\end{aligned}
\end{equation}
\end{definition*}
\end{minipage}\par

This integer-valued expression is invariant under cyclic shifts and overall sign. Proposition~\ref{prop:clean} shows that the minimum shell is nonempty, so define its minimum checker value by
\[
 \mu_E(A):=\min\{E_A(y):y\in\mathcal I,\ \norm{y}_2^2=h+1\}.
\]
The next proposition is the promised intermediate reduction: Gap Exact Set Cover is converted into a gap for $E_A$ restricted to shortest vectors.

\begin{proposition}[Checker gap on the minimum shell]
\label{prop:checkergap}
On promised source inputs,
\[
 \text{YES}\Longrightarrow\mu_E(A)\le\tau,
 \qquad
 \text{NO}\Longrightarrow\mu_E(A)>\tfrac32\tau.
\]
In particular, in the NO case every minimum vector satisfies $E_A(y)\ge\tau+1$.
\end{proposition}
\begin{proof}
\textbf{Soundness for every minimum vector.}
Rotate and, if necessary, negate $y$ so that its anchor is at $0$ and it is binary. This leaves $E_A(y)$ unchanged. Decode $z_i=y_{d_i}+y_{-d_i}$. The displayed repair-cost formula and Lemma~\ref{lem:corr} give
\begin{equation}
\label{eq:checkerlower}
 E_A(y)\ge\Psi_A(z)\ge\OPT(A),
\end{equation}
where the last step is Lemma~\ref{lem:repair}. This includes the possibility $z_i=2$; no Booleanity assumption on the decoded vector is needed. On a source NO instance,
\[
 E_A(y)\ge\OPT(A)>\tfrac32\tau
\]
for every minimum vector. Since $E_A(y)$ is an integer, it is also at least $\tau+1$.

\textbf{Completeness for one YES vector.}
Choose $\xi\in\{0,1\}^r$ with $A\xi=\one_m$ and $\norm{\xi}_1\le\tau$. Proposition~\ref{prop:clean} supplies a clean minimum vector with $y_{d_i}=\xi_i$ and $y_{-d_i}=0$, hence its decoder is $z=\xi$. Cleanliness gives
\[
 C_y(d_i)=\xi_i,\quad C_y(2d_i)=0,\quad
 C_y(d_j-d_i)=\xi_i\xi_j,\quad C_y(d_i+d_j)=0.
\]
Thus every correlation replacement is an equality, and
\[
 E_A(y)=\Psi_A(\xi)
 =\norm{\xi}_1+\norm{A\xi-\one_m}_2^2
 =\norm{\xi}_1\le\tau.
\]
No claim that $\xi$ attains $\OPT(A)$ is needed.
\end{proof}

Proposition~\ref{prop:checkergap} controls only the minimum shell. No checker lower bound is asserted for longer vectors; Section~\ref{sec:conversion} excludes all such vectors directly from the final SVP threshold using the large scalar part of the multiplier.

\section{From checker values to Euclidean lengths}
\label{sec:conversion}
Proposition~\ref{prop:checkergap} gives a gap in the checker value on the minimum shell. We now convert that scalar quadratic statistic into the ordinary Euclidean norm of a new ideal. The conversion has two steps: first represent the checker by a symmetric circulant operator $K$; then add a large scalar multiplier and apply the resulting ring element to the ideal.

\subsection{The checker as a symmetric circulant operator}
Write
\[
 E_A(y)=m+\sum_{d\in\Z/(q-1)\Z}\beta_d C_y(d),\qquad\beta_d\in\Z,
\]
where $\beta_d$ is the coefficient of the listed displacement in~\eqref{eq:checker} and is zero at every other residue. We list one orientation only; $C_y(d)=C_y(-d)$. Every displacement sum in this section is over $\Z/(q-1)\Z$.

Let $S$ be the $(q-1)\times(q-1)$ cyclic shift matrix. Since
\[
 C_y(d)=y^{\mathsf T}S^dy
 =\frac12 y^{\mathsf T}(S^d+S^{-d})y,
\]
the nonconstant part of $E_A$ is already a symmetric circulant quadratic form. The only obstruction is the constant $m$. On the minimum shell, however,
\[
 m=\frac{m}{h+1}\norm{y}_2^2.
\]
Thus the whole checker has a homogeneous quadratic representation there. Multiplying by $2(h+1)$ clears the denominator and motivates the definition
\begin{equation}
\label{eq:K}
 K=2mI_{q-1}+(h+1)\sum_{d\in\Z/(q-1)\Z}\beta_d(S^d+S^{-d}).
\end{equation}
The matrix $K$ is an integer symmetric circulant $(q-1)\times(q-1)$ matrix. For every vector,
\[
 y^{\mathsf T}Ky=2m\norm{y}_2^2+2(h+1)\sum_d\beta_d C_y(d),
\]
and hence
\begin{equation}
\label{eq:Kshell}
 y^{\mathsf T}Ky=2(h+1)E_A(y)
 \qquad\text{when }\norm{y}_2^2=h+1.
\end{equation}
So $K$ is simply the integer symmetric circulant operator representing the checker on the minimum shell; no positivity of $K$ is required.

Because $S$ represents multiplication by $X$, the same operator is multiplication in $R_{q-1}$ by
\[
 k(X)=2m+(h+1)\sum_d\beta_d(X^d+X^{-d}).
\]
Here negative exponents mean the corresponding residues modulo $q-1$, since $X$ is a unit in $R_{q-1}$.

\subsection{From the checker operator to the final ideal}
We next add a large scalar part. Its purpose is twofold: on the minimum shell, the cross term with $K$ will record the checker; away from the shell, the scalar part will keep longer vectors long.

Define
\[
 L=1+\max_i\sum_j|K_{ij}|.
\]
Symmetry makes the maximum absolute column sum equal to the maximum absolute row sum, namely $L-1$. For any real vector $y$, weighted Cauchy--Schwarz gives
\begin{align*}
 \norm{Ky}_2^2
 &\le \sum_i\left(\sum_j|K_{ij}|\right)
               \left(\sum_j|K_{ij}|\,|y_j|^2\right)\\
 &\le (L-1)\sum_j|y_j|^2\sum_i|K_{ij}|
 \le (L-1)^2\norm{y}_2^2.
\end{align*}
Thus $\norm{K}_{\mathrm{op}}\le L-1<L$.

Choose
\begin{equation}
\label{eq:M}
 M=1+\max\{L^2,\ 2L(h+2)+(4\tau+2)(h+1)\}.
\end{equation}
Now define the ring multiplier and the final ideal directly by
\begin{equation}
\label{eq:output}
 g(X)=M+k(X),\qquad I'=g(X)\mathcal I.
\end{equation}
In coefficient coordinates, multiplication by $g(X)$ is the integer symmetric circulant matrix
\[
 T=MI_{q-1}+K,
\]
so equivalently $I'=T\mathcal I$. This is why the image lattice is the natural final object: measuring $y\in\mathcal I$ by the perturbed quantity $\norm{Ty}_2$ is exactly the same as measuring its image $v=Ty\in I'$ by the ordinary Euclidean norm required by SVP.

The ring definition also makes the ideal property immediate. For $a\in R_{q-1}$ and $y\in\mathcal I$,
\[
 a\bigl(g(X)y\bigr)=g(X)(ay)\in g(X)\mathcal I.
\]
Thus $I'$ is exactly the principal multiple $g\mathcal I$, not an enlargement by ideal closure.

The reverse triangle inequality gives
\begin{equation}
\label{eq:Tlower}
 \norm{Ty}_2\ge M\norm{y}_2-\norm{Ky}_2\ge(M-L)\norm{y}_2.
\end{equation}
Since $M>L$, multiplication by $g$, equivalently $T$, is injective on coefficient space. Hence $I'$ is a full-rank integral ideal and every vector of $I'$ has a unique preimage in $\mathcal I$.

The key identity is already visible before choosing the threshold. For $v=Ty$,
\begin{equation}
\label{eq:length}
 \norm{v}_2^2=M^2\norm{y}_2^2+2M y^{\mathsf T}Ky+\norm{Ky}_2^2.
\end{equation}
On the minimum shell,~\eqref{eq:Kshell} turns the cross term into the checker:
\[
 \norm{v}_2^2=(h+1)\bigl(M^2+4M E_A(y)\bigr)+\norm{Ky}_2^2.
\]
Thus the final ideal has exactly the intended meaning: its ordinary Euclidean lengths are a large baseline determined by the old norm, plus a first-order term proportional to the checker, plus the controlled error $\norm{Ky}_2^2$.

\begin{theorem}[Exact gap preservation]
\label{thm:conversion}
Define the integer squared-length threshold
\[
 \boxed{B_*:=(h+1)\bigl(M^2+(4\tau+2)M\bigr).}
\]
Then $I'=g(X)\mathcal I=T\mathcal I$ is a full-rank integral ideal of $R_{q-1}$, and on promised source inputs,
\[
 \text{source YES}\quad\Longleftrightarrow\quad\lambda_1(I')^2\le B_*.
\]
\end{theorem}
\begin{proof}
Equation~\eqref{eq:length} gives, on the minimum shell,
\[
 \norm{Ty}_2^2=(h+1)\bigl(M^2+4M E_A(y)\bigr)+\norm{Ky}_2^2.
\]
The threshold is chosen as the exact midpoint between the first term at checker values $\tau$ and $\tau+1$:
\[
 \frac{(h+1)(M^2+4M\tau)+(h+1)(M^2+4M(\tau+1))}{2}
 =(h+1)(M^2+(4\tau+2)M)=B_*.
\]
This midpoint calculation is not yet a proof of separation: we must control the additional term and also vectors outside the shell. We do so in three cases.

\textbf{A YES vector on the minimum shell.} Proposition~\ref{prop:checkergap} supplies $y$ with $\norm{y}_2^2=h+1$ and $E_A(y)\le\tau$. By~\eqref{eq:Kshell} and $\norm{Ky}_2^2\le L^2(h+1)$,
\[
 \norm{Ty}_2^2
 \le M^2(h+1)+4M(h+1)\tau+L^2(h+1)
 < B_*.
\]
Indeed $L^2(h+1)<2M(h+1)$ since $M>L^2$, so the error is smaller than the distance to the midpoint. Also $Ty\ne0$ by injectivity.

\textbf{Any NO vector on the minimum shell.} Here $E_A(y)\ge\tau+1$. Using~\eqref{eq:Kshell} and $\norm{Ky}_2^2\ge0$,
\[
 \norm{Ty}_2^2\ge M^2(h+1)+4M(h+1)(\tau+1)>B_*.
\]
The margin is at least $2M(h+1)$.

\textbf{Every nonzero vector outside the shell.} Its squared coefficient norm is an integer at least $h+2$. We do not use any checker inequality here. Instead,~\eqref{eq:Tlower} gives
\[
 \norm{Ty}_2^2\ge(M-L)^2(h+2).
\]
Subtracting the threshold,
\begin{align*}
 (M-L)^2(h+2)-B_*
 &=M\bigl[M-2L(h+2)-(4\tau+2)(h+1)\bigr]\\
 &\hspace{10mm}+L^2(h+2)>0
\end{align*}
by~\eqref{eq:M}. This bound holds regardless of the source answer.

The first case supplies a sufficiently short output vector on YES inputs. The other two cases exclude every nonzero output vector on NO inputs.
\end{proof}

Theorem~\ref{thm:conversion} proves correctness on every promised source input. Section~\ref{sec:algorithm} completes Theorem~\ref{thm:main} by constructing the output basis and threshold in deterministic polynomial time and giving the NP verifier.

\Needspace{21\baselineskip}
\section{Basis construction and complexity}
\label{sec:algorithm}
The single modular kernel of Lemma~\ref{lem:baseann} lets us compute the base ideal by standard integer linear algebra. We use Smith normal form with its unimodular transformations, which is computable in deterministic polynomial time~\cite{KB}. A unimodular integer matrix has determinant $\pm1$ and an integer inverse. Basis vectors are columns throughout.

\subsection{A basis of a modular kernel}
\begin{lemma}[Basis from Smith normal form]
\label{lem:kernel}
Let $C\in\Z^{d\times d}$ and let $a\ge2$ be an integer. Compute unimodular integer matrices $U,V$ such that
\[
 UCV=\operatorname{diag}(\sigma_1,\ldots,\sigma_d),
\]
where zero diagonal entries are allowed. Then
\[
 \boxed{B=V\operatorname{diag}\!\left(
 \frac{a}{\gcd(a,\sigma_1)},\ldots,\frac{a}{\gcd(a,\sigma_d)}
 \right)}
\]
is an integer basis of $\{y\in\Z^d:Cy\equiv0\pmod a\}$. Here $\gcd(a,0)=a$.
\end{lemma}
\begin{proof}
Write $y=Vx$. Since $V$ is unimodular, $y$ ranges over all integer vectors exactly when $x$ does. Also $U$ and its inverse preserve divisibility by $a$. Thus
\[
 Cy\equiv0\pmod a
 \iff \sigma_i x_i\equiv0\pmod a\quad(1\le i\le d).
\]
For each $i$, this scalar condition is equivalent to
\[
 \frac{a}{\gcd(a,\sigma_i)}\mid x_i.
\]
When $\sigma_i=0$, the divisor is $1$ and there is no restriction. Otherwise the equivalence follows by dividing by the greatest common divisor and cancelling the remaining coprime factor. The coordinates of $x$ can therefore be chosen independently as multiples of these divisors. Multiplying by $V$ gives exactly the displayed basis, with no extra vectors and none omitted.
\end{proof}

Let $C_P$ be the $(q-1)\times(q-1)$ integer circulant matrix for multiplication by the polynomial $P$ of Lemma~\ref{lem:baseann}. That lemma gives
\[
 \mathcal I=\{y\in\Z^{q-1}:C_Py\equiv0\pmod{q(h^2-1)}\}.
\]
Apply Lemma~\ref{lem:kernel} with $C=C_P$ and $a=q(h^2-1)$ to obtain a basis $\mathbf B_{\mathcal I}$. This uses an integer normal form, not Gaussian elimination over a composite residue ring, and requires no prime factorization of the modulus.

\subsection{Index and output}
The index can be read directly from the original constraints. The map
\[
 y\longmapsto\bigl(y(-h)\bmod(h^2-1),\ y(\alpha),\ldots,y(\alpha^{k-1})\bigr)
\]
from $R_{q-1}$ to $(\Z/(h^2-1)\Z)\times\F_q^{k-1}$ has kernel $\mathcal I$ and is onto. The first component is onto using constant polynomials. The moment evaluations are jointly onto by interpolation at the $k-1$ distinct points, using polynomials of degree at most $k-2<q-1$. Coefficientwise Chinese remaindering combines these prescriptions because $q$ and $h^2-1$ are coprime. Hence
\begin{equation}
\label{eq:index}
 [\Z^{q-1}:\mathcal I]=(h^2-1)q^{k-1}.
\end{equation}
Compute $K,L,M,B_*$ as in Section~\ref{sec:conversion} and output
\[
 \mathbf B_{\rm out}=(MI_{q-1}+K)\mathbf B_{\mathcal I},\qquad B_*.
\]
Its columns generate exactly the image ideal used in Theorem~\ref{thm:conversion}.

\subsection{Polynomial bit complexity}
The parameters satisfy $q-1=O((r+1)^{12})$ and $k=O(r+1)$, with fixed exponents. The NTRU-form application doubles the dimension to $2(q-1)$ and obeys the same asymptotic bound. The coefficient representatives for $P$ and the entries of $C_P$ lie between $0$ and $q(h^2-1)-1$. Thus $C_P$ has polynomial binary encoding length in the explicit source input.

The polynomial-time Smith normal form algorithm computes both the diagonal form and the matrices $U,V$ with polynomial bit length~\cite{KB}. Every factor $a/\gcd(a,\sigma_i)$ in Lemma~\ref{lem:kernel} is at most $a=q(h^2-1)$, so the resulting basis has polynomial bit length as well. This is a bit-complexity assertion, not merely a bound on the number of arithmetic operations.

The offsets are $O(r^3)$. Since $|S_i|,|S_i\cap S_j|\le m$, the sum of absolute checker coefficients is $O(mr^2+mr+r)$. The row sums of $K$ are bounded by
\[
 2m+2(h+1)\sum_{d\in\Z/(q-1)\Z}|\beta_d|.
\]
Hence $K,L,M,B_*$ have polynomial bit length. The final basis multiplication uses polynomially many products and sums on integers of polynomial bit length. Prime selection, primitive-element selection, and the polynomial operations constructing $P$ were already bounded in Sections~\ref{sec:parameters} and~\ref{sec:anchor}. Thus the complete reduction is deterministic polynomial time, although its dimension is far too large for practical use.

No step enumerates the minimum shell or computes a clean completion. The counting theorem is used only to prove that a YES witness has an appropriate output vector.

\subsection{NP membership and the search consequence}
For a general target input of Definition~1.1, validity is checked by exact rational arithmetic. In addition to nonsingularity, test
\[
 B^{-1}SB\text{ has integer entries}.
\]
This tests whether $S$ maps the column lattice into itself. Because the cyclic shift has finite order, that inclusion implies equality. The inverse products and integrality tests have polynomial bit complexity.

A YES certificate is a nonzero ambient integer vector $v$, rather than basis coefficients. Verify
\[
 v\ne0,\qquad \sum_i v_i^2\le b,\qquad B^{-1}v\text{ has integer entries}.
\]
The last test is exactly lattice membership. Each coordinate has absolute value at most $\sqrt b$, so the certificate has polynomial length; all checks are exact. The language therefore belongs to NP. Combining this verifier with the deterministic reduction proves Theorem~\ref{thm:main} for decision-SVP.

An exact search solver returns a shortest nonzero output vector. One call, followed by comparison of its integer squared norm with $B_*$, decides the constructed instance. This proves the stated search hardness. The NP certificate need only exhibit a vector below threshold, not certify that it is globally shortest.

\medskip
\textbf{Scope and further questions.}
The output is a principal multiple $g\mathcal I$, not necessarily a principal ideal. The parity mechanism uses the two evaluations
\[
 y(1)=u(y)+v(y),\qquad y(-1)=u(y)-v(y).
\]
Both belong to the reducible cyclic quotient-ring setting. Projecting to one cyclotomic factor can discard the constraints that produce the anchor. Hardness over full rings of integers in prescribed field families, and principality of the cyclic output, would require additional arguments.

A second question is hardness within a fixed approximation factor greater than one. Our final multiplier is deliberately close, after scaling, to the identity: the checker supplies a strict threshold gap, while a large scalar protects every vector outside the minimum shell. The proof does not give a constant relative separation. Neither the gap in the source cover problem nor tensoring arbitrary cyclic lattices automatically supplies such a result; tensor products naturally carry a product of cyclic actions rather than one full coordinate cycle. A constant-factor cyclic reduction is therefore a separate target, not a corollary claimed here.

\section{Application to the algebraic class of NTRU-form lattices}
\label{sec:ntru}
The cyclic lattices constructed above also give hardness for a rank-two convolution-lattice class. The application uses the specific congruence description of our hard instances; it is not an assertion that every cyclic lattice admits the representation needed below.

\subsection{The target class and the reduction idea}
For $N\ge1$, let $R_N=\Z[X]/(X^N-1)$. Given an integer $Q\ge2$ and a polynomial $H\in R_N/QR_N$, define
\begin{equation}
\label{eq:ntruclass}
 \Lambda_{H,Q}=\{(x,z)\in R_N^2:Hx\equiv z\pmod{QR_N}\}.
\end{equation}
Its norm is the unweighted coefficient Euclidean norm
\[
 \norm{(x,z)}_2^2=\norm{x}_2^2+\norm{z}_2^2.
\]
If $C_H$ is the integer circulant matrix for multiplication by an arbitrary lift of $H$, an explicit \emph{column} basis is
\begin{equation}
\label{eq:ntrubasis}
 \mathbf B_{H,Q}=\begin{pmatrix}I_N&0\\ C_H&QI_N\end{pmatrix}.
\end{equation}
Indeed its column combinations are precisely $(x,C_Hx+Qw)$ with $x,w\in\Z^N$. Changing the lift of $H$ changes $C_H$ by $Q$ times an integer matrix and leaves this lattice unchanged. Thus the lattice has rank $2N$ and determinant $Q^N$. It is invariant under simultaneous cyclic shifts of the two blocks, not necessarily under one cyclic shift of all $2N$ coordinates. This is the usual unweighted NTRU public-lattice form, with the row/column convention fixed explicitly; see Hoffstein--Pipher--Silverman~\cite[Section 3.4.1]{NTRU}.

\begin{definition*}[Algebraic NTRU-form SVP]
The input consists of $N\ge1$, $Q\ge2$, the $N$ coefficients of $H$ modulo $Q$, and a nonnegative integer squared threshold $b$. Decide whether $\lambda_1(\Lambda_{H,Q})^2\le b$. The degree, modulus, and polynomial are unrestricted parts of the input; in particular $H$ need not be a unit modulo $Q$. No secret-key or key-generation promise is imposed.
\end{definition*}

\begin{theorem}[Application to NTRU form]
\label{thm:ntru}
Exact Euclidean decision SVP for the algebraic NTRU-form class is NP-complete under deterministic polynomial-time many-one reductions. Hardness holds with $N=q-1$ for a varying odd prime $q$, so the output dimension is $2(q-1)$. Exact search-SVP on this class is NP-hard under polynomial-time Turing reductions. These statements concern the unrestricted algebraic class, not the cryptographically generated subclass or a distribution of NTRU public keys.
\end{theorem}

For the rest of this section $N=q-1$. By Lemma~\ref{lem:baseann}, the base ideal $\mathcal I$ is already the kernel of multiplication by one polynomial modulo $q(h^2-1)$, and $q(h^2-1)R_N\subseteq\mathcal I$. We now modify the final length-conversion multiplier so that the hard image ideal $J=T_0\mathcal I$ again admits a single modular-kernel description. Once this is established, a simple threshold-preserving NTRU embedding preserves exactly all vectors below the SVP threshold.

\subsection{A congruence-preserving hard image ideal}
We now use the same symmetric circulant checker matrix $K$ and row-sum bound $L$ as in Section~\ref{sec:conversion}, but modify the final multiplier so that it is congruent to the identity modulo $q(h^2-1)$. Define
\begin{equation}
\label{eq:ntrurescale}
\begin{aligned}
 M_0&=1+q(h^2-1)
 \max\{L^2,\ 2L(h+2)+(4\tau+2)(h+1)\},\\
 T_0&=M_0I_N+q(h^2-1)K,\qquad J=T_0\mathcal I,\\
 b_{\rm cyc}&=(h+1)\bigl(M_0^2+(4\tau+2)M_0q(h^2-1)\bigr).
\end{aligned}
\end{equation}
Then
\[
 \boxed{T_0\equiv I_N\pmod{q(h^2-1)}.}
\]

\begin{lemma}[The modified multiplier preserves the gap]
\label{lem:ntrugap}
$J$ is a full-rank cyclic ideal, and on promised source inputs,
\[
 \text{source YES}\quad\Longleftrightarrow\quad
 \lambda_1(J)^2\le b_{\rm cyc}.
\]
Moreover,
\[
 \det(T_0)>1,\qquad \det(T_0)\equiv1\pmod{q(h^2-1)}.
\]
In particular, $\gcd\bigl(\det(T_0),q(h^2-1)\bigr)=1$.
\end{lemma}
\begin{proof}
\textbf{The image is a full-rank ideal.}
Since $T_0$ is an integer polynomial in the cyclic shift $S$, it represents multiplication by an element of $R_N$. Thus $J=T_0\mathcal I$ is an ideal. We next check that the map is injective and that its determinant is positive.

The symmetric matrix $K$ has every eigenvalue in $[-(L-1),L-1]$. Every eigenvalue of $T_0$ is consequently at least
\[
 M_0-q(h^2-1)(L-1)
 \ge1+q(h^2-1)(L^2-L+1)>1.
\]
Here we used $M_0-1\ge q(h^2-1)L^2$. Therefore $T_0$ is nonsingular over $\Q$, $J$ has full rank, and $\det(T_0)>1$. Taking determinants in $T_0\equiv I_N\pmod{q(h^2-1)}$ gives the claimed determinant congruence and coprimality.

\textbf{Vectors on the minimum shell.}
For every $y\in\mathcal I$,
\[
 \norm{T_0y}_2^2
 =M_0^2\norm{y}_2^2
 +2M_0q(h^2-1)y^{\mathsf T}Ky
 +q^2(h^2-1)^2\norm{Ky}_2^2.
\]
When $\norm{y}_2^2=h+1$, identity~\eqref{eq:Kshell} turns this into
\[
 (h+1)\bigl(M_0^2+4M_0q(h^2-1)E_A(y)\bigr)
 +q^2(h^2-1)^2\norm{Ky}_2^2.
\]
For a YES shell vector, $E_A(y)\le\tau$ and
\[
 q^2(h^2-1)^2\norm{Ky}_2^2
 \le q^2(h^2-1)^2L^2(h+1)
 <2M_0q(h^2-1)(h+1),
\]
since $M_0>q(h^2-1)L^2$. This is smaller than the margin to the threshold, so its squared length is below $b_{\rm cyc}$. For a NO shell vector, $E_A(y)\ge\tau+1$ and the last term is nonnegative; its squared length exceeds $b_{\rm cyc}$ by at least $2M_0q(h^2-1)(h+1)$.

\textbf{All other nonzero vectors.}
The minimum squared norm is $h+1$, and squared norms of integer vectors are integers. Hence every nonzero vector outside that shell has $\norm{y}_2^2\ge h+2$. By the reverse triangle inequality,
\begin{align*}
 \norm{T_0y}_2
 &=\norm{M_0y+q(h^2-1)Ky}_2\\
 &\ge M_0\norm{y}_2-q(h^2-1)\norm{Ky}_2\\
 &\ge\bigl(M_0-q(h^2-1)L\bigr)\norm{y}_2.
\end{align*}
The coefficient is positive. Squaring and using $\norm{y}_2^2\ge h+2$ gives
\[
 \norm{T_0y}_2^2\ge\bigl(M_0-q(h^2-1)L\bigr)^2(h+2).
\]
Subtracting the threshold yields
\begin{align*}
 &\bigl(M_0-q(h^2-1)L\bigr)^2(h+2)-b_{\rm cyc}\\
 &\quad=M_0\Bigl[M_0-q(h^2-1)
       \bigl(2L(h+2)+(4\tau+2)(h+1)\bigr)\Bigr]\\
 &\qquad\quad+q^2(h^2-1)^2L^2(h+2).
\end{align*}
The bracket is at least $1$, because the maximum in the definition of $M_0$ is at least $2L(h+2)+(4\tau+2)(h+1)$. Explicitly,
\[
 M_0-q(h^2-1)\bigl(2L(h+2)+(4\tau+2)(h+1)\bigr)\ge1.
\]
The first term in the difference is therefore at least $M_0>0$, and the remaining term is nonnegative. Every off-shell vector maps strictly above $b_{\rm cyc}$. Together with the two shell cases, this proves the decision gap.
\end{proof}

\subsection{A single annihilator for the hard image ideal}
The distinction between rational and integral invertibility matters here. We have proved that $T_0^{-1}$ exists as a matrix over $\Q$. We have \emph{not} inverted its multiplying element inside $R_N$: such an integral inverse would force $\det(T_0)=\pm1$, whereas $\det(T_0)>1$. The proof uses the integral adjugate and checks divisibility before forming any integral preimage.

Since $T_0S=ST_0$, its rational inverse also commutes with $S$. Therefore
\[
 \operatorname{adj}(T_0)=\det(T_0)T_0^{-1}
\]
is an integer matrix commuting with $S$. Such a matrix is circulant: each successive column is the cyclic shift of the preceding column. Let $\widetilde g\in R_N$ be the polynomial whose multiplication matrix is $\operatorname{adj}(T_0)$. The adjugate identity reads
\[
 \operatorname{adj}(T_0)T_0=T_0\operatorname{adj}(T_0)
 =\det(T_0)I_N.
\]
Each adjugate entry is an integer polynomial in the entries of $T_0$. Reducing $T_0\equiv I_N$ modulo $q(h^2-1)$ thus gives
\[
 \operatorname{adj}(T_0)\equiv I_N\pmod{q(h^2-1)},\qquad
 \widetilde g\equiv1\pmod{q(h^2-1)R_N}.
\]
Define
\begin{equation}
\label{eq:transportann}
 \boxed{P_J=\widetilde g+\det(T_0)(P-1).}
\end{equation}
The determinant congruence then implies
\[
 \begin{aligned}
 P_J&\equiv P\pmod{q(h^2-1)R_N},\\
 P_J&\equiv\widetilde g\pmod{\det(T_0)R_N}.
 \end{aligned}
\]
These are the two residues needed to impose base-ideal membership and membership in the image of $T_0$ simultaneously.

\begin{lemma}[A single modular kernel for $J$]
\label{lem:transportann}
Define
\[
 Q_1:=\det(T_0)q(h^2-1).
\]
Then the hard image ideal has the exact description
\begin{equation}
\label{eq:hardann}
 \boxed{J=\{x\in R_N:P_Jx\equiv0\pmod{Q_1R_N}\}.}
\end{equation}
\end{lemma}
\begin{proof}
\textbf{Membership in the image of $T_0$.}
For $x\in R_N\cong\Z^N$,
\[
 x\in T_0R_N
 \iff \operatorname{adj}(T_0)x\in\det(T_0)\Z^N.
\]
If $x=T_0y$ with $y$ integral, the adjugate product equals $\det(T_0)y$. Conversely, if that product is coefficientwise divisible by $\det(T_0)$, then
\[
 y=\frac{\operatorname{adj}(T_0)x}{\det(T_0)}\in\Z^N
 \quad\text{and}\quad T_0y=x.
\]
This is division in $\Q^N$ after checking integer divisibility, not division in a residue ring. Here $T_0R_N$ denotes the image of the multiplication map on $R_N$.

\textbf{Why the relevant intersection is $J$.}
For every $y\in R_N$, the identity congruence gives
\[
 T_0y-y\in q(h^2-1)R_N\subseteq\mathcal I.
\]
If $x=T_0y$ with $y\in\mathcal I$, then $x\in T_0R_N$ and $x=y+(T_0y-y)\in\mathcal I$. Hence $J\subseteq T_0R_N\cap\mathcal I$.

Conversely, if $x\in T_0R_N\cap\mathcal I$, choose an integral $y\in R_N$ with $x=T_0y$. Since $x-y\in\mathcal I$ and $x\in\mathcal I$, also $y=x-(x-y)\in\mathcal I$. Thus $x\in T_0\mathcal I=J$, proving
\[
 \boxed{T_0R_N\cap\mathcal I=T_0\mathcal I=J.}
\]
This equality itself needs only the identity congruence and the displayed containment, not invertibility of $T_0$.

\textbf{Combine the two coprime congruences.}
Because $\det(T_0)$ and $q(h^2-1)$ are coprime, the residues of $P_J$ give
\[
 P_Jx\equiv0\pmod{Q_1R_N}
 \iff
 \begin{cases}
 \operatorname{adj}(T_0)x\in\det(T_0)\Z^N,\\
 Px\equiv0\pmod{q(h^2-1)R_N}.
 \end{cases}
\]
The first condition is $x\in T_0R_N$; the second is $x\in\mathcal I$ by Lemma~\ref{lem:baseann}. Their intersection is exactly $J$, as just proved.
\end{proof}

\subsection{Embedding the hard kernel in NTRU form}
The last step is purely elementary.

\begin{lemma}[Threshold-preserving NTRU embedding]
\label{lem:ntruembedding}
Let $b\in\Z_{\ge0}$. For
\[
 J=\{x\in R_N:P_Jx\equiv0\pmod{Q_1R_N}\},
\]
one has
\[
 \lambda_1(J)^2\le b
 \quad\Longleftrightarrow\quad
 \lambda_1\!\left(\Lambda_{(b+1)P_J,(b+1)Q_1}\right)^2\le b.
\]
\end{lemma}
\begin{proof}
A pair $(x,z)$ belongs to
$\Lambda_{(b+1)P_J,(b+1)Q_1}$ exactly when
\[
 (b+1)P_Jx-z\in(b+1)Q_1R_N.
\]
Hence every coefficient of $z$ is divisible by $b+1$. Write
$z=(b+1)c$. The defining congruence becomes
\[
 c\equiv P_Jx\pmod{Q_1R_N}.
\]
If $c\ne0$, then
\[
 \norm{(x,(b+1)c)}_2^2\ge(b+1)^2>b.
\]
Thus every vector of squared norm at most $b$ has $c=0$, and then the remaining condition is exactly $x\in J$. The map $x\mapsto(x,0)$ preserves nonzeroness and squared length.
\end{proof}

\subsection{Proof of the application and complexity}
\begin{proof}[Proof of Theorem~\ref{thm:ntru}]
Given the promised source instance, compute $P_J$, $Q_1$, and $b_{\rm cyc}$ as above. Define the final lattice directly by
\begin{equation}
\label{eq:ntruoutput}
 \boxed{\Lambda=\left\{(x,z)\in R_N^2:
 (b_{\rm cyc}+1)P_Jx\equiv z
 \pmod{(b_{\rm cyc}+1)Q_1R_N}\right\}.}
\end{equation}
Output a basis of $\Lambda$ with squared-length threshold $b_{\rm cyc}$. Its basis is~\eqref{eq:ntrubasis} with public polynomial $(b_{\rm cyc}+1)P_J$ and modulus $(b_{\rm cyc}+1)Q_1$; the polynomial coefficients are reduced modulo that modulus only when encoding the output.

Lemmas~\ref{lem:ntrugap}, \ref{lem:transportann}, and~\ref{lem:ntruembedding} give
\[
 \text{source YES}
 \iff\lambda_1(J)^2\le b_{\rm cyc}
 \iff\lambda_1(\Lambda)^2\le b_{\rm cyc}.
\]
In fact the complete threshold ball is preserved:
\[
 \{v\in\Lambda:\norm{v}_2^2\le b_{\rm cyc}\}
 =\{(x,0):x\in J,\ \norm{x}_2^2\le b_{\rm cyc}\}.
\]
Thus the argument controls every output vector below the threshold, not only the vector corresponding to a source witness.

All steps are deterministic polynomial time. The bounds in Sections~\ref{sec:anchor} and~\ref{sec:algorithm}, together with the explicit formulas above, give polynomial bit bounds for $P$, $K$, $M_0$, $T_0$, and $b_{\rm cyc}$. Hadamard's inequality gives polynomial bit bounds for $\det(T_0)$ and the adjugate entries, which are computable by exact elimination. The formulas for $P_J$ and $Q_1$ therefore have polynomial bit complexity, as do the scaled public polynomial and modulus in~\eqref{eq:ntruoutput}. Reduction of the polynomial coefficients modulo that modulus and construction of the block basis~\eqref{eq:ntrubasis} are exact integer operations. The modulus may be numerically large; its binary encoding, not its value, must have polynomial size.

For NP membership, a certificate is a nonzero integer pair $(x,z)\in\Z^{2N}$. Check
\[
 \norm{x}_2^2+\norm{z}_2^2\le b,
 \qquad Hx-z\equiv0\pmod{QR_N}.
\]
A pair satisfying the length bound has polynomial bit length, and integer cyclic convolution and the norm comparison are polynomial-time exact operations. This proves decision NP-completeness. One call to an exact search oracle followed by a squared-norm comparison gives the search consequence.
\end{proof}

\subsection{What the application does not establish}
In the constructed instance, both the public polynomial and the modulus contain the factor $b_{\rm cyc}+1>1$. In particular,
\[
 Q_1\bigl((b_{\rm cyc}+1)P_J\bigr)=0
 \quad\text{modulo }(b_{\rm cyc}+1)Q_1R_N,
\]
while $Q_1$ is nonzero modulo that modulus. Thus the public polynomial is not a unit. This common factor is what forces a short vector's second block to vanish.

The theorem concerns the unrestricted algebraic class. It imposes neither the small-secret and invertibility conditions of cryptographic NTRU key generation~\cite[Sections 1.1--1.2]{NTRU} nor a distribution of public keys. Noninvertibility alone is not a test of cryptographic validity. In particular, this is not an average-case security theorem or hardness of secret-key recovery; reductions for specified search-NTRU variants~\cite{PMS} concern different promises.

\Needspace{8\baselineskip}
\phantomsection

\bigskip
\noindent\textbf{Daqing Wan}\\
Center for Discrete Mathematics and College of Mathematics and Statistics,\\
Chongqing University, Chongqing 401331, China.\\
\texttt{dwan@math.uci.edu}

\end{document}